\documentclass[10pt,aps,prx,twocolumn,longbibliography,superscriptaddress,nofootinbib]{revtex4-2}

\usepackage{amsmath, empheq,braket,amsthm,amssymb,graphics,amsfonts,array,color, mathrsfs}
\usepackage[justification=raggedright, singlelinecheck=false]{caption}
\usepackage{subcaption}
\usepackage{tikz}
\usetikzlibrary{shapes.geometric, arrows.meta, positioning, calc, backgrounds, fit}

\usepackage{booktabs}
\usepackage{thmtools}
\usepackage{thm-restate}
\usepackage{ytableau}
\declaretheorem[name=Theorem]{Theo}
\declaretheorem[name=Lemma]{Lemma}

\usepackage{bm}
\usepackage[unicode=true,pdfusetitle, bookmarks=true,bookmarksnumbered=false,bookmarksopen=false,breaklinks=false,pdfborder={0 0 0},backref=false,colorlinks=true,linkcolor=blue,citecolor=blue,urlcolor=blue]{hyperref}

\usepackage{xcolor}

\usepackage{graphicx,lipsum}
\usepackage{algorithm}
\usepackage{algpseudocode}
 \makeatletter
  \renewcommand{\ALG@name}{Protocol}
  \makeatother

\newcommand{\kb}[2]{\left| #1 \vphantom{#2} \right>\left< #2 \vphantom{#1} \right|} 

\newcommand{\tr}[1]{\mathrm{Tr}\left[ #1 \right]}
\newcommand{\dd}{\mathrm{d}}

\newcommand{\n}{\mathbf{n}} 

\renewcommand{\H}{\mathcal{H}} 
\newcommand{\Psym}{\mathbb{P}_{sym}}
\newcommand{\Pc}{\mathbb{P}_{C}}
\newcommand{\Sym}{\mathrm{Sym}} 
\newcommand{\D}{\mathcal{D}} 

\newcommand{\dtr}{d_{\mathrm{tr}} }
\newcommand{\Find}{F_{\mathrm{ind}} }

\DeclareMathOperator{\perm}{perm}
\DeclareMathOperator{\imm}{imm}
\DeclareMathOperator{\fix}{fix}

\newtheorem{prop}{Proposition}
\newtheorem{rem}{Remark}

\newtheorem{defin}{Definition}

\newcommand{\equalcontrib}{
\thanks{These authors contributed equally to this work.\\
\href{marco.robbio@ulb.be}{marco.robbio@ulb.be}\\
\href{oszmaniec@cft.edu.pl}{oszmaniec@cft.edu.pl}\\
}}

\DeclareRobustCommand{\permrev}[2]{#2}

\begin{document}


\title{Optimal interferometric certification of multi-photon indistinguishability}

\author{Marco Robbio}
\equalcontrib
\affiliation{Centre for Quantum Information and Communication, \'Ecole polytechnique de Bruxelles, CP 165/59, Universit\'e libre de Bruxelles, 1050 Brussels, Belgium}
\affiliation{International Iberian Nanotechnology Laboratory (INL), Av. Mestre Jos\'e Veiga, 4715-330 Braga, Portugal}

\author{Micha\l \ Oszmaniec }
\equalcontrib
\affiliation{Center for Quantum Enabled-Computing, Center for Theoretical Physics of the Polish Academy of Sciences, Al. Lotników 32/46, 02-668 Warsaw, Poland}

\author{Nicolas J. Cerf}
\affiliation{Centre for Quantum Information and Communication, \'Ecole polytechnique de Bruxelles, CP 165/59, Universit\'e libre de Bruxelles, 1050 Brussels, Belgium}

\author{Ernesto Galvão}
\affiliation{International Iberian Nanotechnology Laboratory (INL), Av. Mestre Jos\'e Veiga, 4715-330 Braga, Portugal}
\affiliation{Instituto de F\'isica, Universidade Federal Fluminense, Av. Gal. Milton Tavares de Souza s/n, Niter\'oi, RJ, 24210-340, Brazil}

\author{Leonardo Novo}
\affiliation{International Iberian Nanotechnology Laboratory (INL), Av. Mestre Jos\'e Veiga, 4715-330 Braga, Portugal}


\begin{abstract}
    Multiphoton indistinguishability is a key resource for photonic quantum technologies, yet its characterization typically relies on resource-intensive methods. In this work, we develop two efficient and experimentally friendly protocols to estimate or bound the fidelity $F_{\mathrm{ind}}$ of an $N$-photon state to the closest perfectly indistinguishable state. The first protocol applies to sources preparing separable states, uses a single Fourier interferometer together with photon-number-resolving detection, and yields tight two-sided bounds on $F_{\mathrm{ind}}$. The second protocol combines randomized implementations of linear-optical interferometers with photon counting, enabling direct estimation of $F_{\mathrm{ind}}$ for arbitrary $N$-photon states. 
    Both protocols can certify $F_{\mathrm{ind}} = 1 - \mathcal{O}(\epsilon)$ using provably optimal $\mathcal{O}(1/\epsilon)$ samples, in contrast to previous approaches which required prior assumptions on the model of partial distinguishability.
    Our methods, based on a multiphoton generalization of the Hong-Ou-Mandel test, bring the rigorous and operationally meaningful certification of multiphoton indistinguishability within reach of current photonic technologies.
\end{abstract}

\maketitle

\section{Introduction}\label{sec:introduction}

Indistinguishability of multiple photons is a defining ingredient of quantum many-body interference and underpins a broad range of photonic quantum technologies, including quantum communication and networking \cite{Pirandola2015,luMultiparticleEntanglementSwapping2009,hoEntanglementCommunicationComplexity2022}, quantum cryptography \cite{chenQuantumSecretSharing2005,proiettiExperimentalQCKA2021,pickstonConferenceKeyAgreement2023,rueckleAnonymousConference2023,khurana2024foundingquantumcryptographyquantum}, quantum metrology \cite{nagataBeatingStandardQuantum2007,xiangOptimalMultiphotonPhase2013}, proposals for photonic quantum computational advantage \cite{aaronsonComputationalComplexityLinear2010,Bouland2026PRX,wangBosonSampling20Photons2019,zhongQuantumComputationalAdvantage2020,madsenQuantumComputationalAdvantage2022}, and photonic quantum computing \cite{Knill2001,shiHighfidelityPhotonicQuantum2022,Bartolucci2023,maringVersatileSinglePhoton2024}. In practice, photons are never perfectly identical because they may differ in internal degrees of freedom, such as polarization, frequency, or arrival time. Partial distinguishability degrades many-body interference \cite{tichy2014interference,Schesnovic19} and, consequently, the performance of linear-optical quantum gates \cite{rohdeModeMismatch2006}, quantum-networking primitives \cite{coganDeterministicGeneration2023}, and photonic sampling experiments \cite{renemaClassSim,zhongQuantumComputationalAdvantage2020,wangHighEfficiencyBosonSampling2017}.

For two photons, partial distinguishability interpolates between fully distinguishable particles, for which interference is absent, and perfectly indistinguishable bosons, as exemplified by the Hong--Ou--Mandel (HOM) effect \cite{hongMeasurementSubpicosecondTime1987,Bouchard2020}. For \(N\) photons, however, a complete description requires genuinely multiparticle quantities that go beyond this two-photon picture, including collective photonic phases \cite{Tichy2015,Shchesnovich2015,Menssen_17,Agne_17, shchesnovich2018collective, jonesInterferingDistinguishablePhotons2020a, annoniIncoherentBehaviorPartially2025, rodari2026experimental}. For product internal states, corresponding to independent photon sources, these effects are encoded in multivariate traces \(\tr{\rho_{i_1}\rho_{i_2}\cdots\rho_{i_k}}\) of the internal states \(\{\rho_1,\ldots,\rho_N\}\) \cite{shchesnovich2018collective}. These quantities, also known as higher-order \emph{Bargmann invariants}, arise more generally in quantum theory as basis-independent descriptors of relational properties among quantum states \cite{Oszmaniec2024,pancharatnam1956generalized,simon1993bargmann,chruscinski2004geometric,bamber2014observing,kirkwood1933quantum,liang2023unified,wagner2023simple,wagner2024quantum}.

A complete characterization of the internal photonic states could, in principle, proceed by determining all associated Bargmann invariants, although their number grows super-exponentially with \(N\). This motivates the search for experimentally accessible and operationally meaningful measures of multiphoton indistinguishability that avoid this overhead. Several such approaches have been developed theoretically and experimentally \cite{stanisic2018discriminating, Brod_19, somhorst2023quantum, englbrechtIndistinguishabilityIdenticalBosons2024,annoniIncoherentBehaviorPartially2025, rodari2024}.  Focusing on the measure of $N$-photon indistinguishability introduced in \cite{Brod_19},  Ref.~\cite{Pont_22} introduced and implemented a photonic circuit for estimating this quantity, albeit with a statistical cost that grows exponentially with \(N\). Very recently, approaches based on Fourier interference have shown that this exponential sampling overhead can be avoided \cite{sanz2026exponentialimprovementbenchmarkingmultiphoton,schadow2026certificationlinearopticalquantum}. However, these works rely on additional structural assumptions on the distinguishability noise based on an incoherent, partition-based description of distinguishability, but do not directly characterize arbitrary coherent multiphoton internal states.

\begin{figure*}[hbt!]
    \centering
    
    \includegraphics[width=\textwidth]{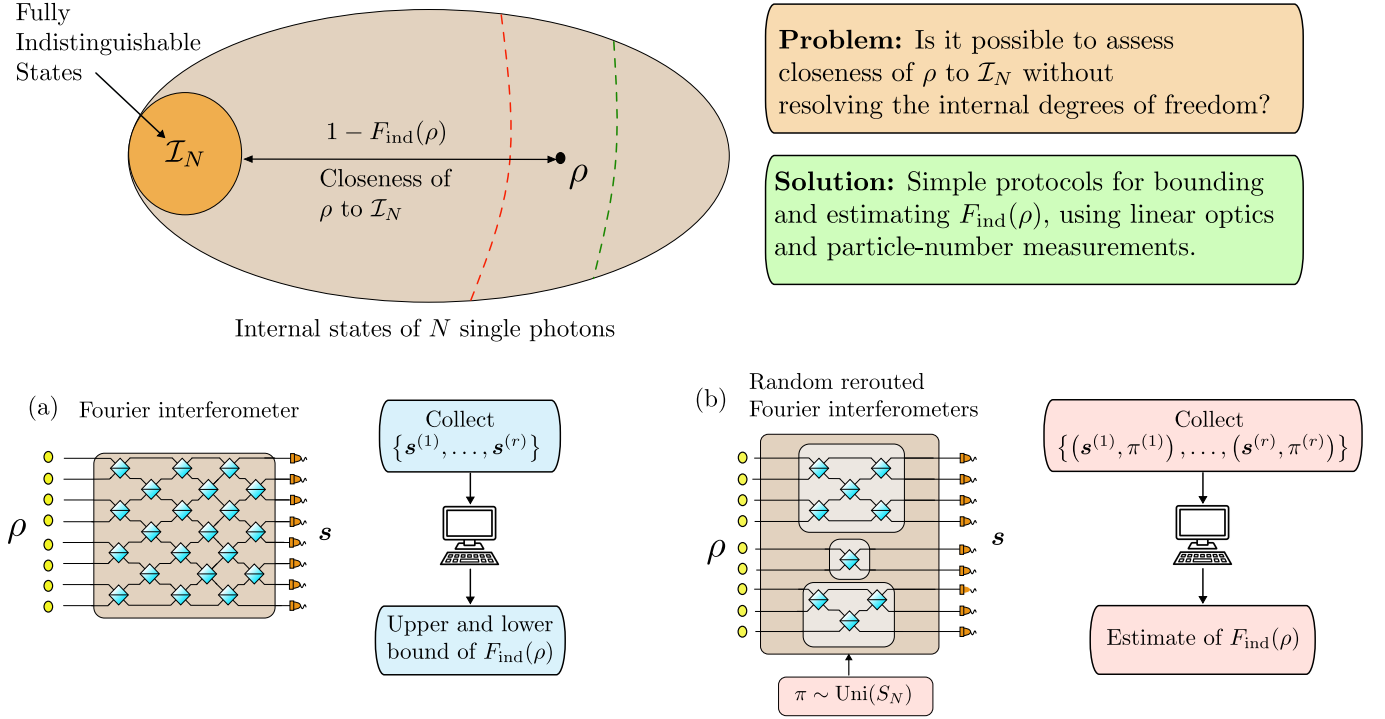}
    \caption{We consider an unknown state $\rho$ describing internal states of an $N$-photon input with exactly one photon in each of $N$ spatial modes. We seek  to quantify its closeness to the set of states describing perfectly indistinguishable photon states, without resolving the internal degrees of freedom. This closeness is measured by the fidelity $\Find(\rho)$ to the closest perfectly indistinguishable state. \textbf{Top:} Geometric representation of the certification problem. \textbf{Bottom left (a):} The cyclic protocol uses a single fixed Fourier interferometer followed by photon-number-resolving detection. Post-processing the observed output patterns provides upper and lower bounds (for independent photon sources) on $\Find(\rho)$. \textbf{Bottom right (b):} The randomized protocol  samples a random permutation $\pi\in S_N$ and implements the corresponding rerouted Fourier interferometer diagonalizing $\pi$. Post processing of the output allows to directly estimate $\Find(\rho)$.}
    \label{fig: Protocols}
\end{figure*}

In this work, we propose operationally motivated and experimentally practical methods for quantifying multiphoton indistinguishability in an \(N\)-photon state with one photon in each occupied input mode. We first show that, for an arbitrary collective internal state \(\rho\), the fidelity to the closest perfectly indistinguishable state, \(\Find(\rho)\), is exactly the expectation value of the projector onto the symmetric subspace of the internal degrees of freedom. This quantity is in one-to-one relation with the trace distance between the accessible external state and the ideal collision-free bosonic state, and it therefore uniformly bounds the deviation of the statistics of arbitrary measurements insensitive to internal degrees of freedom from  their ideal values.

We then develop two linear-optical protocols; see Fig.~\ref{fig: Protocols}. The experimentally simpler protocol applies to separable internal states, as produced by classically correlated  independent photon sources. It uses a single \(N\)-mode Fourier interferometer to estimate the cyclic-symmetric weight \(P_c(\rho)\)  \cite{novo2026nativelinearopticalprotocolefficient}, which gives dimension-independent upper and lower bounds on \(\Find(\rho)\). The second protocol applies to arbitrary, including correlated or entangled, internal states. It directly estimates \(\Find(\rho)\) by sampling permutations and implementing the corresponding collection of Fourier interferometers on disjoint subsets of modes. The state classes covered by these protocols are substantially broader than those of the recent Fourier-based approaches discussed above; Section~\ref{sec: alternative} provides a detailed comparison. 

In the near-perfect indistinguishability regime, both protocols exploit the vanishing variance of their estimators and certify \(\Find(\rho)=1-\mathcal{O}(\varepsilon)\) with \(\mathcal{O}(\varepsilon^{-1}\log\delta^{-1})\) samples at confidence \(1-\delta\). We show that this scaling is information-theoretically optimal, even when the competing states are restricted to pure product states.

The remainder of the paper is organized as follows. Section~\ref{sec: partial distinguishability} introduces the formal description of partially distinguishable photons and the embedding of their internal states into the physical bosonic Hilbert space. Section~\ref{sec: closest ind state} establishes the operational characterization of \(\Find(\rho)\) and formulates the  cyclic-symmetry bounds underlying the single-interferometer protocol. Section~\ref{sec: protocols} presents the two interferometric protocols and their sample-complexity guarantees. Section~\ref{sec: alternative} analyzes the performance of our single-interferometer protocol for  various models of partial distinguishability, and compares our methods with previous characterizations of multiphoton indistinguishability. We conclude in Section~\ref{sec: conclusion}; technical proofs and further results are deferred to the Appendices~\ref{sec: appendix}.

\section{Theory of partially distinguishable photons}\label{sec: partial distinguishability}

We consider the standard interferometric setting for partially distinguishable photons \cite{Tichy2015,Shchesnovich2015}: \(N\) photons enter distinct input modes of an \(M\)-mode interferometer, with \(M\geq N\). A creation operator \(\hat a^\dagger_{j,\alpha}\) carries an external, or path, index \(j\in[M]\), on which the interferometer acts, and an internal index \(\alpha\), which is not resolved by the interferometer or the detectors. Following Refs.~\cite{Tichy2015,englbrechtIndistinguishabilityIdenticalBosons2024}, we denote the corresponding Hilbert spaces by \(\H_{\mathrm{ext}}\) and \(\H_{\mathrm{int}}\). The \(N\)-photon Hilbert space is
\(\H^{(N)}_{\mathrm{bos}}\simeq\Sym^N(\H_{\mathrm{ext}}\otimes\H_{\mathrm{int}})\).
Internal degrees of freedom may be discrete, such as polarization, continuous, such as frequency or arrival time, or a combination of both. An \(M\times M\) unitary matrix \(U\) induces the linear-optical transformation (independent of $\alpha$)
\begin{equation}\label{eq:linear_interference}
    \hat U\hat a^\dagger_{j,\alpha}\hat U^\dagger
    =\sum_{k=1}^{M}U_{k j}\hat a^\dagger_{k,\alpha},
    \qquad j\in[M].
\end{equation}

The accessible measurements are photon-number measurements that do not resolve the internal degrees of freedom. Their observables are \(\hat N_j=\sum_\alpha \hat n_{j,\alpha}\), where \(\hat n_{j,\alpha}=\hat a^\dagger_{j,\alpha}\hat a_{j,\alpha}\). We denote a joint detection outcome by \(\mathbf{s}=(s_1,\ldots,s_M)\).

In this paper we focus on collision-free inputs with one photon in each of the first \(N\) external modes. A representative pure product state is
\begin{equation}\label{eq:prod2ndQuant}
    \ket{\Omega}
    =\prod_{j=1}^{N}\hat a^\dagger_{j,\phi_j}\ket{\mathrm{vac}},
\end{equation}
where \(\hat a^\dagger_{j,\phi_j}=\sum_\alpha c_{j,\alpha}\hat a^\dagger_{j,\alpha}\) creates a photon with normalized internal state \(\ket{\phi_j}=\sum_\alpha c_{j,\alpha}\ket{\alpha}\in\H_{\mathrm{int}}\) in external mode \(j\). More generally, the collective internal state may be mixed and may contain arbitrary classical or quantum correlations, including entanglement between photons.

It is convenient to express this fixed-occupation sector in first quantization, as in Refs.~\cite{englbrechtIndistinguishabilityIdenticalBosons2024,steinmetz2024simulating,brunner2019many,oszmaniecClassicalSimulationPhotonic2018,brodClassicalSimulationLinear2020}. Let
\(\ket{\mathbf{i}_0}=\bigotimes_{j=1}^{N}\ket{j}\in\H_{\mathrm{ext}}^{\otimes N}\)
and let \({\permrev{\hat P_\tau^{\mathrm{int}}}{\Pi_\tau^{\mathrm{int}}}}\) and \({\permrev{\hat P_\tau^{\mathrm{ext}}}{\Pi_\tau^{\mathrm{ext}}}}\) permute the corresponding tensor factors according to \(\tau\in S_N\). The bosonic embedding \(V:\H_{\mathrm{int}}^{\otimes N}\rightarrow\H_{\mathrm{bos}}^{(N)}\) is
\begin{equation}\label{eq:bosonic_embedding}
    V\ket{\psi}
    =\frac{1}{\sqrt{N!}}\sum_{\tau\in S_N}
    \left({\permrev{\hat P_\tau^{\mathrm{int}}}{\Pi_\tau^{\mathrm{int}}}}\otimes{\permrev{\hat P_\tau^{\mathrm{ext}}}{\Pi_\tau^{\mathrm{ext}}}}\right)
    \left(\ket{\psi}\otimes\ket{\mathbf{i}_0}\right).
\end{equation}
For the product vector \(\ket{\psi}=\bigotimes_{j=1}^{N}\ket{\phi_j}\), Eq.~\eqref{eq:bosonic_embedding} is precisely Eq.~\eqref{eq:prod2ndQuant} written in first quantization. As shown in Lemma~\ref{lem:isometry} of Appendix~\ref{sec: bosemb}, \(V\) is an isometry. We will denote
\(\Omega(\rho)\coloneq V\rho V^\dagger\)
for the bosonic embedding of an arbitrary internal state \(\rho\in\D(\H_{\mathrm{int}}^{\otimes N})\).

Because both the interferometers and the detectors are insensitive to the internal degrees of freedom, all accessible statistics are determined by the reduced external state \cite{englbrechtIndistinguishabilityIdenticalBosons2024,brunner2019many,steinmetz2024simulating}
\begin{equation}\label{eq:external_reduced_state}
    \Omega_{\mathrm{ext}}(\rho)
    \coloneq\mathrm{Tr}_{\mathrm{int}}\!\left[\Omega(\rho)\right].
\end{equation}
The ideal collision-free state of \(N\) indistinguishable photons has occupation vector
\(\mathbf{n}_0=(1^N,0^{M-N})\).
In first quantization it is
\begin{equation}\label{eq:ideal_collision_free_state}
    \ket{\mathbf{n}_0}
    =\frac{1}{\sqrt{N!}}\sum_{{\sigma}\in S_N}
    {\permrev{\hat P_\pi^{\mathrm{ext}}}{\Pi_\sigma^{\mathrm{ext}}}}\ket{\mathbf{i}_0},
\end{equation}
equivalently \(\ket{\mathbf{n}_0}=\prod_{j=1}^{N}\hat a^\dagger_j\ket{\mathrm{vac}}\) in second quantization.

\begin{defin}[Perfectly indistinguishable state]\label{def:perfect_indistinguishability}
The state \(\Omega(\rho)\) describes perfectly indistinguishable photons if and only if
\begin{equation}\label{eq:perfectINDIST}
    \Omega_{\mathrm{ext}}(\rho)
    =\kb{\mathbf{n}_0}{\mathbf{n}_0}.
\end{equation}
\end{defin}

Thus, from the perspective of any process and measurement acting only on the external degrees of freedom, \(\Omega(\rho)\) is operationally indistinguishable from the ideal collision-free bosonic state. In realistic settings this exact condition must be relaxed. We quantify closeness using the Uhlmann fidelity \(F(\rho,\sigma)\coloneq\|\sqrt{\rho}\sqrt{\sigma}\|_1^2\) and the trace distance \(\dtr(\rho,\sigma)\coloneq\frac{1}{2}\|\rho-\sigma\|_1\); see Ref.~\cite{NielsenChuang2010QCQI} for their operational interpretation.

\section{Closest indistinguishable states and certification of perfect indistinguishability}\label{sec: closest ind state}
 
Crucially, the reduced state $\Omega_{ext}(\rho)$ is, in general, not bosonic and should instead be regarded as a state on $\H_{ext}^{\otimes N}$. For a general internal state $\rho$, it takes the form
\begin{equation}\label{eq:reducedstate}
\Omega_{ext}(\rho) = \frac{1}{N!} \sum_{\tau,\tau' \in S_N}
{\permrev{\hat{P}^{ext}_\tau}{\Pi^{\mathrm{ext}}_\tau}} \ket{\mathbf{i}_0}\bra{\mathbf{i}_0} {\permrev{\hat{P}^{ext}_{\tau'}}{\Pi^{\mathrm{ext}}_{\tau'}}}\ \tr{ \rho {\permrev{\hat{P}^{int}_{\tau' \tau^{-1}}}{\Pi^{\mathrm{int}}_{\tau' \tau^{-1}}}}} \ .
\end{equation}
Hence, characterizing $\Omega_{ext}(\rho)$ requires the quantities $\tr{ \rho {\permrev{\hat{P}^{int}_{\sigma}}{\Pi^{\mathrm{int}}_{\sigma}}}}$, $\sigma \in S_N$, which also appear in \cite{shchesnovich2014partial} and are referred to as generalized indistinguishabilities in \cite{annoniIncoherentBehaviorPartially2025}. For product states $\rho$, the decomposition of permutations into disjoint cycles expresses these quantities as products of multivariate traces, which can be estimated via Fourier interferometry \cite{novo2026nativelinearopticalprotocolefficient}. As discussed in Appendix \ref{Sec: number of BI}, the number of independent multivariate traces required to determine all generalized indistinguishabilities scales in general as $\mathcal{O}((N-1)!)$. For pure product states $\rho=\psi_1 \otimes \ldots \otimes \psi_N$, however, it reduces to only $\mathcal{O}(N^2)$ parameters for generic configurations of individual states $\psi_i$ \cite{Oszmaniec2024}. 

Rather than aiming at a full description of $\Omega_{ext}(\rho)$, our goal is to quantify how close the input state is to one describing perfectly indistinguishable photons. From \eqref{eq:reducedstate}, we see that
$\Omega_{ext}(\rho)=\ket{\mathbf{n}_0}\bra{\mathbf{n}_0}$
if and only if $\tr{\rho\,{\permrev{\hat{P}^{int}_\tau}{\Pi^{\mathrm{int}}_\tau}}}=1$ for every permutation $\tau\in S_N$. Equivalently, $\tr{\rho\,\Psym^{\mathrm{int}}}=1$ where $
    \Psym^{\mathrm{int}}= \frac{1}{N!} \sum_{\tau\in S_N} {\permrev{\hat{P}^{int}_{\tau}}{\Pi^{\mathrm{int}}_{\tau}}} $ is the projector onto the symmetric subspace in $\H_{int}^{\otimes N}$. Thus, the internal states describing perfectly indistinguishable photons are precisely those supported on $\Sym^N(\H_{int})$. This observation naturally motivates quantifying partial indistinguishability by the fidelity to the closest state with this property.

\begin{defin}[Fidelity to the closest indistinguishable state]\label{def:fidelityIND}
Let $\Omega(\rho)$ be a state of $N$ single photons occupying the first $N$ external modes of the interferometer, with internal states described by $\rho\in \mathcal{D}(\H_{int}^{\otimes N})$. The fidelity to the closest indistinguishable state of $N$ single photons is defined as
\begin{equation}\label{eq:FIDindist}
\Find(\rho)=\max_{\sigma \in \D\left(\Sym^N(\H_{int})\right)} F\left(\Omega(\rho),\Omega(\sigma)\right).
\end{equation} 
\end{defin}

\noindent By construction, $\Find(\rho)$ measures the fidelity of the input state $\Omega(\rho)$ to the set of photonic states describing perfectly indistinguishable photons. The following Lemma shows that this operational quantity has a particularly simple form:

\begin{restatable}{Lemma}{Fidtoclose}\label{lem:closestIND}
The fidelity to the closest state of $N$ indistinguishable single photons (Definition~\ref{def:fidelityIND}) is given by 
\begin{equation}\label{eq:opINTERPpsym}
\Find(\rho) = \tr{\rho\ \Psym^{\mathrm{int}}}\ .
\end{equation}
Additionally, if $\tr{\rho \Psym^{\mathrm{int}}}>0$ the state  $ \tilde{\rho} = \Psym^{\mathrm{int}} \rho \Psym^{\mathrm{int}} / \tr{\rho \Psym^{\mathrm{int}}} $
 satisfies $F(\Omega(\rho),\Omega(\tilde{\rho}))=\Find(\rho)$.
\end{restatable}
This result (independently obtained in \cite{schadow2026certificationlinearopticalquantum}) gives a clear meaning to  the expectation value $\tr{\rho\,\Pi^{\mathrm{int}}_{sym}}$ -- it exactly quantifies the fidelity of the physical (photonic) state $\Omega(\rho)$ with the \emph{closest} state of perfectly indistinguishable photons. 
Going further, we can connect $\Find(\rho)$ to the trace distance between $\Omega_{ext}(\rho)$ and a state of perfectly indistinguishable photons $\kb{\n_0}{\n_0}$.
\begin{restatable}{Lemma}{LemDist}\label{lem:tracedistance}
 The state describing the external degrees of freedom of $N$ photons can be written as
\begin{equation}\label{eq:decREDUCED} 
	\Omega_{ext}(\rho)= \Find(\rho) \kb{\n_0}{\n_0} + (1-\Find(\rho)) \sigma^\perp\ ,
\end{equation}
where $\sigma^\perp \in \D(\H_{ext}^{\otimes N})$ is a state with orthogonal support to  $ \kb{\n_0}{\n_0}$. Consequently, 
\begin{equation}\label{eq:dtrFIdrel}
   \dtr (\Omega_{ext}(\rho),\kb{\n_0}{\n_0} )= 1 - \Find(\rho). 
\end{equation}  
\end{restatable}

For any POVM \(\mathcal{M}=\{M_a\}_{a=1}^{\ell}\) on the external degrees of freedom, let \(p_a=\tr{M_a\Omega_{\mathrm{ext}}(\rho)}\) and \(q_a=\tr{M_a\kb{\mathbf{n}_0}{\mathbf{n}_0}}\). Contractivity of the trace distance and \eqref{eq:dtrFIdrel} give $\mathrm{TVD}(p,q)\leq 1-\Find(\rho)$. Thus, $ \Find(\rho)\geq1-\varepsilon $ guarantees that the outcome distribution of \emph{every} experiment insensitive to internal degrees of freedom differs from its ideal indistinguishable-photon distribution by at most \(\varepsilon\) in total-variation distance.

Given these findings, it is natural to investigate reliable methods to estimate or bound $\Find(\rho)$ experimentally. The protocols realizing this are given in Section \ref{sec: protocols}.

We now derive a certification bound on $\Find(\rho)$ that depends only on cyclic permutation symmetry and is a basis for the protocol utilizing a single Fourier interferometer. Let \(C=(1\,2\,\cdots\,N)\) be the standard cycle of length $N$. Our convention is
${\permrev{\hat P_C^{\mathrm{int}}}{\Pi_C^{\mathrm{int}}}}\ket{\psi_1}\otimes\ldots \otimes\ket{\psi_N}=\ket{\psi_N}\otimes \ket{\psi_1}\ldots \otimes\ket{\psi_{N-1}}$. Let $ \Pc^{\mathrm{int}}
    =\frac{1}{N}\sum_{k=0}^{N-1}{\permrev{\hat P_{C^k}^{\mathrm{int}}}{\Pi_{C^k}^{\mathrm{int}}}}$ be a projector onto subspace of $\H_{int}^{\otimes N}$ that is invariant under a cyclic shift.
We then define \(P_c(\rho)=\tr{\Pc^{\mathrm{int}}\rho}\) to be the weight of \(\rho\) in this subspace. The following theorem shows that $P_c(\rho)$ can be used to establish double-sided bounds on $F_{ind}(\rho)$.

\begin{restatable}[Cyclic-symmetry bounds]{Theo}{CyclicTheo}\label{th: Cyclic theorem}
Let \(\rho\in\D(\H_{\mathrm{int}}^{\otimes N})\). Then we have 
    \begin{equation}\label{eq:cyclic_general_upper_bound}
        \Find(\rho)\leq P_c(\rho).
    \end{equation}
Additionally, if \(\rho\) is  separable, i.e.,
    \(\rho=\sum_a w_a\bigotimes_{i=1}^{N}\rho_i^{(a)}\)
    for a probability distribution \(\{w_a\}\), then
    \begin{equation}\label{eq:cyclic_main_two_sided_bound}
        \max\!\left\{0,\,2P_c(\rho)-1\right\}
        \leq\Find(\rho)\ .
    \end{equation}

\end{restatable}

We now give a sketch of the proof (the full argument is given in Appendix~\ref{sec: Cyclic projector bound}). The upper bound has a direct geometric origin: every fully symmetric vector is cyclically invariant, so the symmetric subspace is contained in the cyclic-invariant subspace and
\(\Psym^{\mathrm{int}}\leq\Pc^{\mathrm{int}}\).
The lower bound uses the additional structure of separable states. For pure product states, the expectations of the cyclic shift can be expressed via Bargmann invariants and factor into products of overlaps around the cycle. A large cyclic-symmetric weight constrains these overlaps and forces a non-negligible fully symmetric component. 

Two aspects of  Theorem~\ref{th: Cyclic theorem} require an emphasis. First, the upper bound \eqref{eq:cyclic_general_upper_bound} is valid without any source model and remains applicable to arbitrary internal states $\rho$ (also highly correlated and entangled). Second, the lower bound \eqref{eq:cyclic_main_two_sided_bound} covers arbitrary convex mixtures of product states, including mixed single-photon states from independent sources, without assumptions on the dimension of \(\H_{\mathrm{int}}\), the wave-packet shape, or a particular distinguishability-noise model. 

In Appendix~\ref{sec: Better bound} we derive tighter lower bounds for product internal states, $\rho=\bigotimes_{i=1}^N\rho_i$, based on testing the cyclic symmetry. In particular  Eq.~\eqref{eq: bound slightly optimal} gives a nonlinear lower bound depending only on \(P_c(\rho)\) which for \(P_c(\rho)=1-\varepsilon\) with \(\varepsilon\ll1\) yields
\begin{equation}\label{eq:near_perfect_cyclic_refinement}
    1-\varepsilon
    \geq\Find(\rho)
    \geq1-\varepsilon
    -\frac{N-2}{2(N-1)}\varepsilon^2
    +\mathcal{O}(\varepsilon^3)\ .
\end{equation}
Hence, near perfect indistinguishability for product states, \(P_c(\rho)\) determines \(\Find(\rho)\) up to a second-order correction.  Crucially, the quantities can be obtained by simple processing of outcomes of a single Fourier interferometer.

The next section shows  how  \(P_c(\rho)\) and \(\Find(\rho)\)  can be estimated interferometrically.

\section{Certification of multiphoton indistinguishability via interferometry}\label{sec: protocols}

The expectation values of operators $\Pc^{\mathrm{int}}, \Psym^{\mathrm{int}}$ on internal state $\rho$ are experimentally accessible from simple measurements of $\Omega_{ext}(\rho)$. This is because, on collision-free input states characterized by occupation pattern $\mathbf{n}_0=(1^N,0^{M-N})$, an external permutation of first $N$ occupied modes implements the inverse permutation of the internal tensor factors, as observed in \cite{Geller2025,schadow2026certificationlinearopticalquantum}. Specifically, Lemma~\ref{lemma: ext to int} proven in Appendix \ref{sec: bosemb} gives
\begin{equation}\label{eq:int_vs_ext_perm}
    \tr{{\permrev{\hat P_\pi^{\mathrm{int}}}{\Pi_\sigma^{\mathrm{int}}}}\rho}
    =\tr{{\permrev{\hat P_{\pi^{-1}}^{\mathrm{ext}}}{\Pi_{\sigma^{-1}}^{\mathrm{ext}}}}\Omega_{\mathrm{ext}}(\rho)}.
\end{equation}
Averaging this identity over the cyclic group or the full symmetric group yields, respectively, 
\begin{align}
    P_c(\rho)
    &=\tr{\Pc^{\mathrm{ext}}\Omega_{\mathrm{ext}}(\rho)},
    \label{eq:cyclic_external_expectation}\\
    \Find(\rho)
    &=\tr{\Psym^{\mathrm{ext}}\Omega_{\mathrm{ext}}(\rho)}.
    \label{eq:symmetric_external_expectation}
\end{align}
Here \(\Pc^{\mathrm{ext}}\) and \(\Psym^{\mathrm{ext}}\) are defined by the same group averages as their internal counterparts. These identities lead to two complementary certification strategies, both relying on particle number resolving detection. Protocol~\ref{Prot: Fourier} uses a fixed Fourier interferometer to estimate $P_c(\rho)$, and for separable internal states it converts the estimated cyclic weight $\widehat{P}_c(\rho)$ into strong bounds on \(\Find(\rho)\). Protocol~\ref{Prot: Randomized} instead requires reconfiguring the interferometer between experimental shots to estimate \(\Find(\rho)\) directly for arbitrary internal states. In both cases the number of state preparations required for a prescribed additive accuracy is independent of the number of photons \(N\). Subsections~\ref{sec: cyclic protocol} and \ref{sec: randomized protocol} derive the two protocols and their sample complexity, while Subsection~\ref{sec: near perfect certification} shows how to use them to certify indistinguishability, with the emphasis on the near-perfect indistinguishability regime in which $\Find(\rho)=1-\epsilon$, $\epsilon\ll 1$.

\subsection{Single-interferometer cyclic protocol}\label{sec: cyclic protocol}

The cyclic protocol exploits the fact that all representations of  powers of the cycle ${\permrev{\hat{P}_{C^k}}{\Pi_{C^k}}}$ have a common eigenbasis. Labeling the occupied modes by \(j\in \{0,\ldots,N-1\}\), the corresponding mode-permutation operator is diagonalized by the \(N\)-mode Fourier interferometer:
\begin{equation}\label{eq:cycle_fourier_diagonalization}
    {\permrev{\hat P_C^{\mathrm{ext}}}{\Pi_C^{\mathrm{ext}}}}
    =
    \hat F_N\hat D\hat F_N^\dagger,
    \qquad
    \hat D
    =
    \exp\!\left(
    \frac{2\pi i}{N}
    \sum_{j=0}^{N-1}j\hat n_j
    \right).
\end{equation}
It is possible to see that for photon-number outcome \(\mathbf{s}=(s_0,\ldots,s_{N-1})\) the corresponding eigenvalue of \(\hat D\) is \(\exp[2\pi i f(\mathbf{s})/N]\), where  $f(\mathbf{s}) = \sum_{j=0}^{N-1}j s_j     \pmod N$.  Let \(p_F(\mathbf{s})\) be the output distribution obtained by applying \(\hat F_N^\dagger\) and measuring photon occupations $\mathbf{s}$. As shown in \cite{novo2026nativelinearopticalprotocolefficient}, we have that
\begin{align}
    \mathbb{E}_{\mathbf{s}\sim p_F}
    \!\left[\delta_{f(\mathbf{s}),0}\right]
    &=
    \frac{1}{N}
    \sum_{k=0}^{N-1}
    \tr{\hat D^k\hat F_N^\dagger
    \Omega_{\mathrm{ext}}(\rho)\hat F_N}
    \notag\\
    &=
    \tr{\Pc^{\mathrm{ext}}
    \Omega_{\mathrm{ext}}(\rho)}
    =
    P_c(\rho).
    \label{eq:cyclic_indicator_identity}
\end{align}
Hence \(P_c(\rho)\) is exactly the probability that the detected configuration satisfies the Fourier suppression condition \(f(\mathbf{s})=0\). The preceding derivation justifies the following protocol for estimation of \(P_c(\rho)\). 

\begin{algorithm}[H]
\caption{Estimating the cyclic weight \(P_c(\rho)\)}
\label{Prot: Fourier}
\begin{algorithmic}
\State \textbf{Input:} external state \(\Omega_{\mathrm{ext}}(\rho)\); number of samples \(r\)
\For{\(m=1,\ldots,r\)}
    \State Apply \(\hat F_N^\dagger\), measure photon numbers, and record \(\mathbf{s}^{(m)}\)
    \State Set \(X_m\gets\delta_{f(\mathbf{s}^{(m)}),0}\)
\EndFor
\State \textbf{Return} \(\widehat P_c=r^{-1}\sum_{m=1}^{r}X_m\)
\end{algorithmic}
\end{algorithm}

Protocol~\ref{Prot: Fourier} is therefore realized by a fixed scattering experiment: the same interferometer $F_N^\dagger$ is used in every trial, and the postprocessing consists only of evaluating one modular sum on measurement outcome $\mathbf{s}$. The following result describes the sufficient number of samples for an estimation of $P_c(\rho)$ to a set accuracy $\eta>0$ that is independent of $N$. 

\begin{Theo}\label{th: cyclic protocol complexity}
Let \(\widehat P_c\) be the estimator returned by Protocol~\ref{Prot: Fourier}. For \(0<\eta,\delta<1\),
\begin{equation}\label{eq:cyclic_sample_complexity_main}
    \Pr\!\left(
    \left|\widehat P_c-P_c(\rho)\right|\leq\eta
    \right)
    \geq1-\delta
\end{equation}
whenever
\begin{equation}\label{eq:cyclic_sample_bound_main}
    r
    \geq
    \left[
    \frac{2P_c(\rho)\left[1-P_c(\rho)\right]}{\eta^2}
    +
    \frac{2}{3\eta}
    \right]
    \ln\!\frac{2}{\delta}.
\end{equation}
\end{Theo}

The proof follows from Bernstein's inequality (stated in Appendix~\ref{sec: estimator concentration}). Specifically, single-shot estimators $X_m$ in Protocol \ref{Prot: Fourier} are Bernoulli random variables  with variance \(\operatorname{Var}(X_m)=P_c(\rho)[1-P_c(\rho)]\).

\subsection{Randomized protocol for direct fidelity estimation}\label{sec: randomized protocol}

The simplicity of the cyclic protocol comes with a tradeoff: for general internal states it measures \(P_c(\rho)\), rather than \(\Find(\rho)\) itself, and a nontrivial lower bound requires separability. To estimate \(\Find(\rho)\) directly without a source-model assumption, we retain the full permutation average in Eq.~\eqref{eq:symmetric_external_expectation}:
\begin{align}
    \Find(\rho)
    &=
    \tr{\Psym^{\mathrm{int}}\rho}
    \notag\\
    &=
    \mathbb{E}_{{\sigma}\sim\operatorname{Unif}(S_N)}
    \!\left[
    \tr{{\permrev{\hat P_\pi^{\mathrm{ext}}}{\Pi_\sigma^{\mathrm{ext}}}}
    \Omega_{\mathrm{ext}}(\rho)}
    \right].
    \label{eq: Random idea}
\end{align}
A sampled permutation \({\sigma}\) decomposes uniquely into disjoint cycles. A cycle of length \(\ell\) is diagonalized by an \(\ell\)-mode Fourier interferometer; consequently, the full mode permutation admits the linear-optical diagonalization
\begin{equation}\label{eq:permutation_optical_diagonalization}
    {\permrev{\hat P_\pi^{\mathrm{ext}}}{\Pi_\sigma^{\mathrm{ext}}}}
    =
    \hat W_{\sigma}\hat D_{\sigma}\hat W_{\sigma}^\dagger,
    \qquad
    \hat D_{\sigma}
    =
    \exp\!\left(
    i\sum_{j=1}^{N}
    \theta_j^{({\sigma})}\hat n_j
    \right),
\end{equation}
where \(\hat W_{\sigma}\) is a direct sum of Fourier interferometers acting on the disjoint cycles of \({\sigma}\), together with the corresponding mode routing. For a photon-number outcome \(\mathbf{s}\) sampled after applying \(\hat W_{\sigma}^\dagger\), define the single-shot score
\begin{equation}\label{eq:randomized_single_shot_score}
    Y_{\sigma}(\mathbf{s})
    =
    \cos\!\left(
    \sum_{j=1}^{N}
    \theta_j^{({\sigma})}s_j
    \right)
    \in[-1,1].
\end{equation}
Using reasoning analogous to that given in subsection \ref{sec: cyclic protocol} we get that, for a fixed ${\sigma}\in S_{N}$, the conditional expectation value is
\(\mathbb{E}[Y_{\sigma}|{\sigma}]
=\Re\left\{\tr{{\permrev{\hat P_\pi^{\mathrm{ext}}}{\Pi_\sigma^{\mathrm{ext}}}}\Omega_{\mathrm{ext}}(\rho)}\right\}=\Re\{\tr{{\permrev{\hat P_\pi^{\mathrm{ext}}}{\Pi_\sigma^{\mathrm{ext}}}}\Omega_{\mathrm{ext}}(\rho)}\}\).
Averaging over the sampled permutation therefore gives \(\Find(\rho)\). The preceding derivation justifies the following protocol for estimating \(\Find(\rho)\).

\begin{algorithm}[H]
\caption{Directly estimating \(\Find(\rho)\)}
\label{Prot: Randomized}
\begin{algorithmic}
\State \textbf{Input:} external state \(\Omega_{\mathrm{ext}}(\rho)\); number of samples \(r\)
\For{\(m=1,\ldots,r\)}
    \State Sample \({\sigma}_m\) uniformly from \(S_N\)
    \State Construct \(\hat W_{{\sigma}_m}\) and phases
    \(\{\theta_j^{({\sigma}_m)}\}\) from the
    \Statex \hspace{\algorithmicindent} disjoint-cycle decomposition of \({\sigma}_m\)
    \State Apply \(\hat W_{{\sigma}_m}^\dagger\), measure photon numbers, and record \(\mathbf{s}^{(m)}\)
    \State Set \(Y_m\gets Y_{{\sigma}_m}(\mathbf{s}^{(m)})\)
\EndFor
\State \textbf{Return} \(\widehat F_{\mathrm{ind}}=r^{-1}\sum_{m=1}^{r}Y_m\)
\end{algorithmic}
\end{algorithm}

The required experimental setting relies on linear optics and particle number detectors, just like for Protocol \ref{Prot: Fourier}. Its additional experimental cost is fast reconfiguration  due to the need of applying randomly changing unitary $\hat{W}^\dagger_{\sigma}$ with each trial. Note however that for moderate $N$ it is possible simply to measure each expectation value $\tr{{\permrev{\hat{P}_\pi}{\Pi_\sigma}} \Omega_{\mathrm{ext}}(\rho)}$ directly for every ${\sigma}\in S_N$. For the relevant case of identical, independent sources ($\rho = \rho_0^{\otimes N}$), this experimental overhead can be bypassed entirely: in Appendix~\ref{sec: iid estimation}, we show that $\Find(\rho)$ can be estimated directly without randomization by measuring only $N-1$ power traces on a collection of Fourier interferometers. The following result shows that analogously to the Fourier-based protocol,  the sufficient number of samples the protocol requires to estimate $\Find(\rho)$ to a set accuracy $\eta>0$ is independent of $N$. 

\begin{restatable}{Theo}{SymmetricTheo}\label{th: symmetric Fidelity complexity}
Let \(\widehat F_{\mathrm{ind}}\) be the estimator returned by Protocol~\ref{Prot: Randomized}. For \(0<\eta,\delta<1\),
\begin{equation}\label{eq:randomized_sample_complexity_main}
    \Pr\!\left(
    \left|\widehat F_{\mathrm{ind}}-\Find(\rho)\right|
    \leq\eta
    \right)
    \geq1-\delta
\end{equation}
whenever
\begin{equation}\label{eq:randomized_sample_bound_main}
    r
    \geq
    \left[
    \frac{2\left[1-\Find^2(\rho)\right]}{\eta^2}
    +
    \frac{4}{3\eta}
    \right]
    \ln\!\frac{2}{\delta}.
\end{equation}
\end{restatable}

The proof is a straightforward application of Bernstein's inequality (c.f. Appendix~\ref{sec: estimator concentration}) together with observation  that  \(\operatorname{Var}(Y_m)\leq1-\Find^2(\rho)\)  (resulting from the fact that single-shot estimators $Y_m$ take values in $[-1,1]$).

\subsection{Source certification from finite data}\label{sec: near perfect certification}

We now turn the protocol outcomes into confidence intervals for the indistinguishability fidelity of an unknown source. Suppose that the cyclic and randomized protocols produce the estimates $\widehat P_c$ and $\widehat F_{\mathrm{ind}}$, respectively. Fix a common estimation accuracy $\eta>0$ and choose the sample numbers according to Eqs.~\eqref{eq:cyclic_sample_bound_main} and~\eqref{eq:randomized_sample_bound_main}. For an arbitrary internal state, the randomized protocol gives \footnote{For readability, we leave implicit the intersection of each confidence interval with the physical range $[0,1]$.}
\begin{equation}\label{eq:randomized_apriori_interval_main}
    \Find(\rho)\in
    \left[\widehat F_{\mathrm{ind}}-\eta,\,
    \widehat F_{\mathrm{ind}}+\eta\right]
\end{equation}
with probability at least $1-\delta$. For a separable internal state, the cyclic protocol and Theorem~\ref{th: Cyclic theorem} give
\begin{equation}\label{eq:cyclic_apriori_fidelity_interval_main}
    \Find(\rho)\in
    \left[2\widehat P_c-1-2\eta,\,
    \widehat P_c+\eta\right]
\end{equation}
with the same confidence. Without separability, the cyclic protocol still guarantees the upper bound $\Find(\rho)\leq\widehat P_c+\eta$. For independent sources, $\rho=\bigotimes_{i=1}^N\rho_i$, the refinement of Theorem~\ref{th: Cyclic theorem} discussed above in connection with Eq.~\eqref{eq:near_perfect_cyclic_refinement} tightens this interval near perfect indistinguishability. Proposition~\ref{prop: refined cyclic confidence} in Appendix~\ref{sec: Better bound} gives the precise interval. When $1-\widehat{P_c}+\eta \ll 1$, its lower endpoint asymptotically reduces to $\widehat P_c-\eta$, mirroring the form of the bound from the randomized protocol.

Crucially, Eqs.~\eqref{eq:randomized_apriori_interval_main} and~\eqref{eq:cyclic_apriori_fidelity_interval_main} enable a precise assessment of the indistinguishability of $N$ photons generated by a source, without directly accessing their internal degrees of freedom or assuming their dimension, while allowing for correlations in $\rho$ and using a number of samples independent of the photon number $N$. The sample-size requirements in Eqs.~\eqref{eq:cyclic_sample_bound_main} and~\eqref{eq:randomized_sample_bound_main} depend on the unknown values $P_c(\rho)$ and $\Find(\rho)$. Without detailed information about the source, one may use worst-case variance bounds, at the cost of the usual $\mathcal{O}(\eta^{-2})$ sample complexity. Alternatively, empirical Bernstein bounds can replace the unknown variances with quantities evaluated from the collected data~\cite{maurer2009empirical}. Additionally, we note that  bounds ~\eqref{eq:randomized_apriori_interval_main} and~\eqref{eq:cyclic_apriori_fidelity_interval_main} can be used for formal two sided certification tests (in a sense discussed, e.g., in Ref.~\cite{Kliesch_2021}).

We finally specialize to the regime of high indistinguishability, 
$\Find(\rho)=1-\varepsilon$ with $\varepsilon\ll1$, and ask for additive resolution $\eta=\kappa\varepsilon$, where $\kappa>0$ is constant. Since $\Find(\rho)\leq P_c(\rho)$, the single-shot variances obey $\operatorname{Var}(X_m)=P_c(\rho)[1-P_c(\rho)]\leq\varepsilon$ and $\operatorname{Var}(Y_m)\leq1-\Find^2(\rho)\leq2\varepsilon$.
Substituting these into  Eqs.~\eqref{eq:cyclic_sample_bound_main} and~\eqref{eq:randomized_sample_bound_main} shows that both protocols resolve the indistinguishability on scale $\mathcal{O}(\varepsilon)$ using only  $\mathcal{O}(\varepsilon^{-1}\log\delta^{-1})$ samples. In Appendix~\ref{sec: sample optimality} we prove that this scaling cannot be improved in general, even for pure product inputs and arbitrary collective measurements on all available copies.

\section{Comparison to previous methods}\label{sec: alternative}
While our work focuses on estimating and bounding the fidelity to indistinguishable states, $F_{\mathrm{ind}}(\rho)$, it is useful to compare this quantity with previous approaches to quantifying multiphoton indistinguishability. In the following two subsections, we discuss its relation to two notions in particular: genuine $N$-photon indistinguishability introduced in Ref.~\cite{Brod_19}, and the measure based on the projector $\mathbb{P}^{\mathrm{ext}}_{sym}$ onto the symmetric subspace of the external degrees of freedom introduced in Ref.~\cite{englbrechtIndistinguishabilityIdenticalBosons2024}.

\subsection{Genuine indistinguishability}
The notion of genuine $N$-photon indistinguishability was originally introduced in Ref.~\cite{Brod_19} for internal states that are diagonal in a product basis $\ket{\bm{\alpha}}=\bigotimes_{i=1}^N \ket{\alpha_i}$ built from a fixed orthonormal basis $\lbrace \ket{\alpha}\rbrace$ of the Hilbert space $\H_{\mathrm{int}}$:\footnote{Equation~\eqref{eq:diag state} need not hold for the actual internal state $\rho$. It suffices that $\Omega_{\mathrm{ext}}(\rho)=\Omega_{\mathrm{ext}}(\sigma)$ for some $\sigma$ of this form. By  equation ~\eqref{eq:reducedstate} this holds if and only  all expectation values of permutations evaluated on $\rho$ and $\sigma$ match.}
\begin{equation}\label{eq:diag state}
    \rho = \sum_{\bm{\alpha}} p_{\bm{\alpha}} \ket{\bm{\alpha}}\bra{\bm{\alpha}}\ .
\end{equation}
Reorganizing this sum according to the pattern of internal labels gives
\begin{equation}\label{eq:pos_partition}
    \rho= c_N \rho^{\parallel} + \sum_{k} c_k \rho^\perp_k
\end{equation}
Here, $c_N+\sum_k c_k=1$ and $c_N,c_k\geq0$, with $c_N$ denoting the coefficient associated with genuine $N$-photon indistinguishability. The state $\rho^{\parallel}=\sum_\alpha q_\alpha (\ket{\alpha}\bra{\alpha})^{\otimes N}$ is the component in which all photons occupy the same internal state. Each $\rho_k^\perp$ is a mixture of configurations with the same partition of the photon labels into at least two groups: photons within each group share a pure internal state, while states associated with distinct groups are orthogonal. The index $k$ runs over all such partitions. With these conditions, Eq.~\eqref{eq:pos_partition} describes precisely the positive partition states, up to the operational equivalence specified above \cite{annoniIncoherentBehaviorPartially2025,schadow2026certificationlinearopticalquantum}.

Very recently, Sanz \emph{et al.} \cite{sanz2026exponentialimprovementbenchmarkingmultiphoton} showed that, for input states admitting the partition structure in Eq.~\eqref{eq:pos_partition}, the coefficient $c_N$ can be estimated using $\mathcal{O}(\epsilon^{-2})$ samples of a Fourier interferometer, assuming number of photons $N$ is prime and $\mathcal{O}(\mathrm{poly}(N)/\epsilon^{2})$ samples for general $N$. Independently, Schadow \emph{et al.}~\cite{schadow2026certificationlinearopticalquantum}, following up on previous work \cite{somhorst2023quantum},  derived for the same class of input states lower bounds on $c_N$ based on analogous measurements of $P_c(\rho)$. These bounds also provide lower bounds on $F_{\mathrm{ind}}(\rho)$, since under the positive-partition assumption one has $F_{\mathrm{ind}}(\rho)\geq c_N$.

In connection with our work, we make three remarks. First, assuming the structure in Eq.~\eqref{eq:pos_partition}, one has
\begin{equation}\label{eq:cn_cycle}
    c_N = \tr{\Pi_C^{\text{int}}\rho}= \tr{\Pi_{C^{-1}}^{\text{ext}}\Omega_{ext}(\rho)},  
\end{equation}
for collision-free inputs with a single photon per occupied mode, by Eq.~\eqref{eq:int_vs_ext_perm}. Indeed, by definition of the states $\rho_k^\perp$, $\tr{\Pi_C^{\text{int}}\rho_k^\perp}=0$. Consequently, $c_N$ can be extracted directly by postprocessing the samples from a Fourier-interference experiment \cite{novo2026nativelinearopticalprotocolefficient}, since the Fourier interferometer diagonalizes the cyclic shift acting on the external modes. This gives a conceptually simpler route to measuring $c_N$ than obtaining or bounding it from measurements of $P_c(\rho)$, which requires additional steps and depends on considerations such as whether the photon number is prime \cite{sanz2026exponentialimprovementbenchmarkingmultiphoton,schadow2026certificationlinearopticalquantum}.

Our second remark concerns the positive-partition assumption itself. To justify a decomposition of the form in Eq.~\eqref{eq:pos_partition}, Refs.~\cite{schadow2026certificationlinearopticalquantum, sanz2026exponentialimprovementbenchmarkingmultiphoton} argue that averaging over permutations of the external modes can be used to transform input states of partially distinguishable photons into such an incoherent mixture with a positive $c_N$ coefficient. We show in Appendix~\ref{sec:negative_partition} that this is not always possible by giving an explicit family of pure product inputs for which the coefficient $c_N$ after permutation twirling is negative for every $N\geq3$. Thus, initially uncorrelated photons might not admit a positive partition representation even after mode-permutation twirling, and Eq.~\eqref{eq:pos_partition} remains an additional assumption.

Third, the decomposition $\Omega_{ext}(\rho)= \Find(\rho) \kb{\n_0}{\n_0} + (1-\Find(\rho)) \sigma^{\perp}$ from Lemma \ref{lem:tracedistance} holds for every internal state of a collision-free input, separating the ideal external state from a positive remainder with orthogonal support, without any additional assumptions. Our fidelity-based framework therefore has a substantially broader scope. The quantity $F_{\mathrm{ind}}(\rho)$ is well defined and can be estimated by the randomized protocol even for entangled and potentially correlated internal states, while the single-interferometer lower bound requires only separability, not a positive partition representation.

 \subsection{Bounding \texorpdfstring{\(\Find(\rho)\)}{the indistinguishability fidelity} from HOM visibilities}

\begin{figure*}[tbhp!]
   \centering
    \begin{subfigure}[b]{0.45\textwidth}
        \centering
        \caption{$X$ model}
        \includegraphics[width=\textwidth]{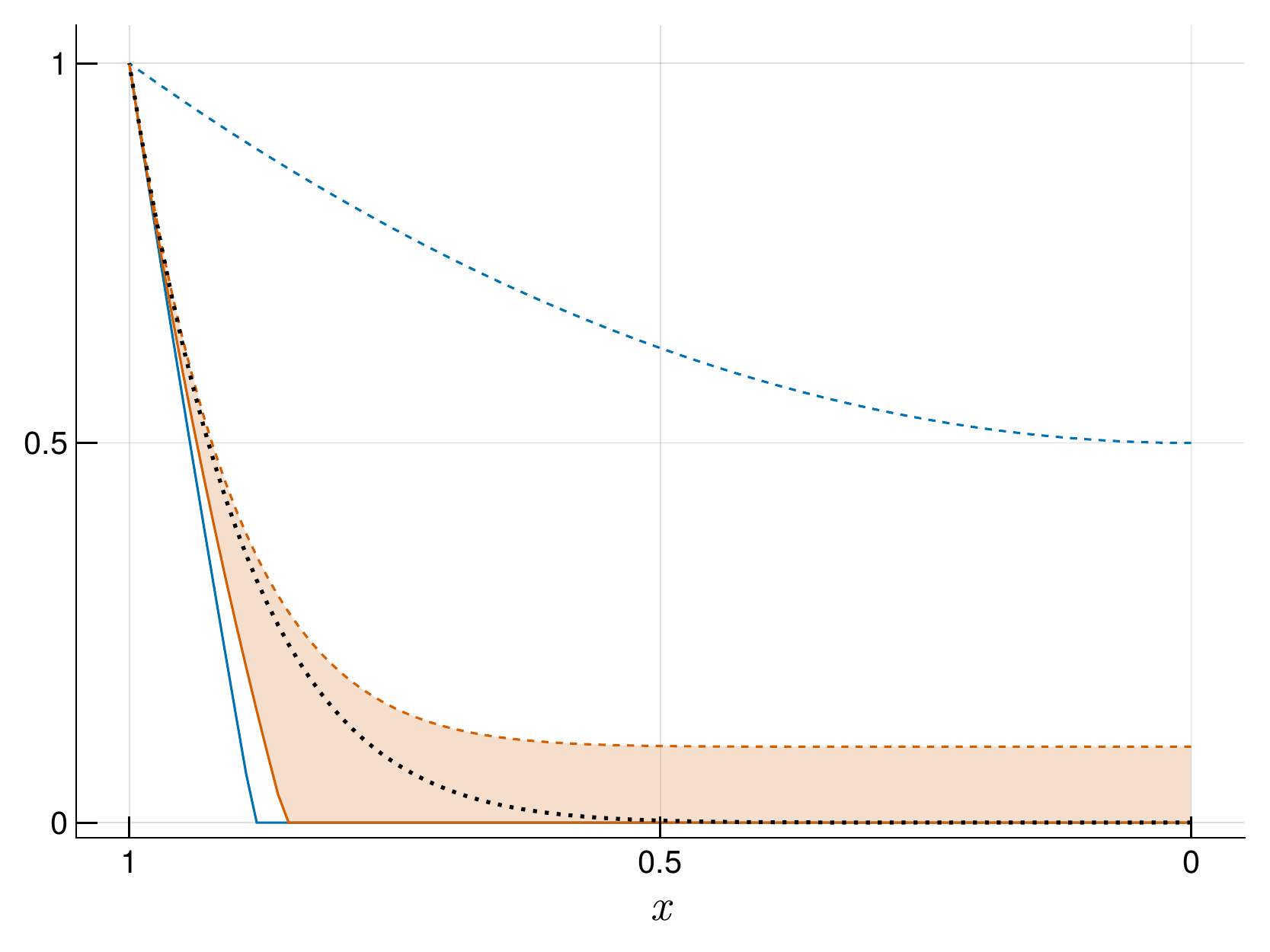}
       
        \label{fig:comparison-x-model}
    \end{subfigure}
    \hfill
    \begin{subfigure}[b]{0.45\textwidth}
        \centering
        \caption{Random time delay}
        \includegraphics[width=\textwidth]{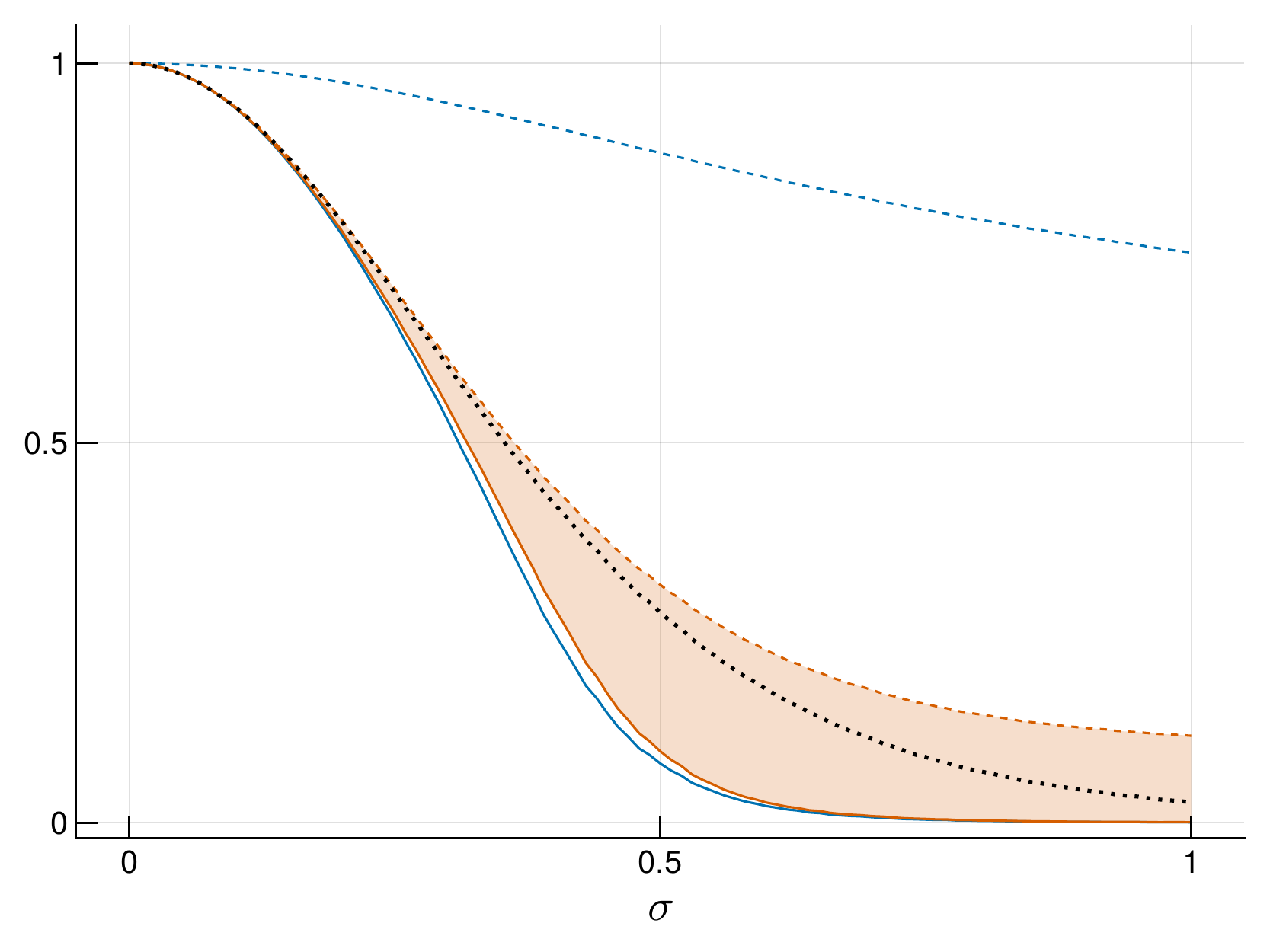}
        
        \label{fig:comparison-time-delay}
    \end{subfigure}

    \vfill

    \begin{subfigure}[b]{0.45\textwidth}
        \centering
        \caption{Random rotations}
        \includegraphics[width=\textwidth]{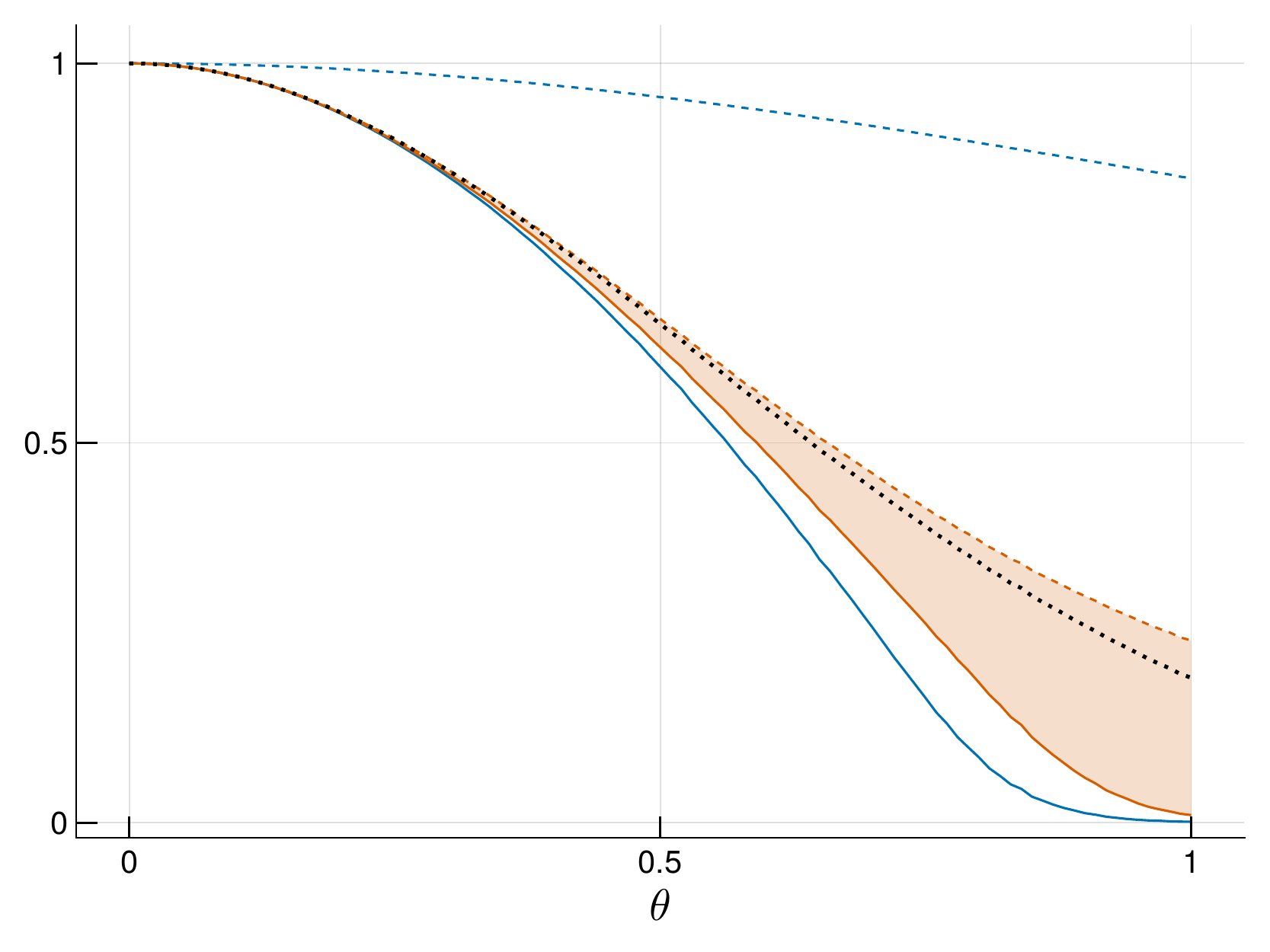}
        
        \label{fig:comparison-rotations}
    \end{subfigure}
    \hfill
    \begin{subfigure}[b]{0.45\textwidth}
        \centering
        \caption{Bad-batch model}
        \includegraphics[width=\textwidth]{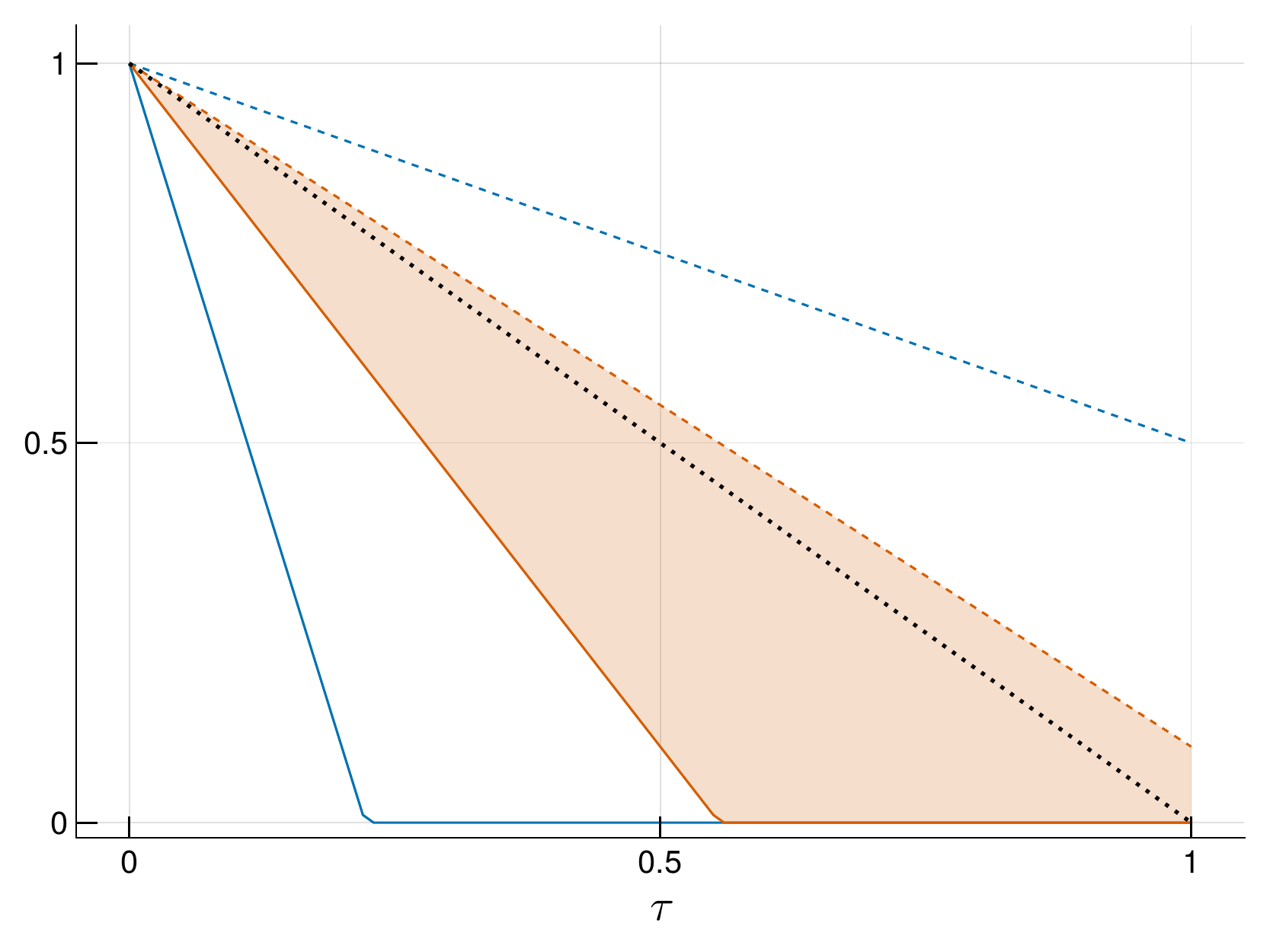}
        
        \label{fig:comparison-bad-batch}
    \end{subfigure}

    \vfill

    \includegraphics[width=0.95\textwidth]{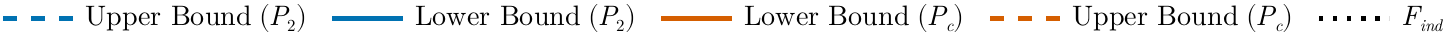}

    \caption{\textbf{Comparison of certification bounds across representative noise models.}
    All panels show $N=10$ photons. The indistinguishability fidelity $\Find$ is compared
    with bounds inferred from the cyclic weight $P_c$ and from the average pairwise
    Hong--Ou--Mandel statistic $P_2$. \textbf{(a) $X$ model:} Pure internal states
    have a uniform pairwise overlap $\langle\phi_i|\phi_j\rangle=x$ for $i\neq j$,
    interpolating between mutually orthogonal and identical states.
    \textbf{(b) Random time delay:} Gaussian temporal wave packets acquire random
    arrival-time offsets, modeling timing jitter.
    \textbf{(c) Random rotations:} Initially identical pure internal states undergo
    independent random unitary rotations, modeling coherent internal-state mismatch.
    \textbf{(d) Bad-batch model:} A mixture of indistinguishable particles and fully distinguishable one, with respective probability $(\tau,1-\tau)$.}
    \label{fig:Comparison}
\end{figure*}

Englbrecht \emph{et al.}~\cite{englbrechtIndistinguishabilityIdenticalBosons2024} introduced the weight of the fully symmetric component of the external state $\tr{\mathbb{P}^{\mathrm{ext}}_{sym} \Omega_\mathrm{ext}(\rho)}$ as a measure of the truly bosonic character of the state, and showed that this quantity is conserved under operations that do not affect the internal degrees of freedom. This measure is not, in general, the same quantity as the indistinguishability fidelity $F_{\mathrm{ind}}(\rho)$ considered here. However, for the collision-free inputs studied in this work, with at most one photon per input mode, Eq.~\eqref{eq:int_vs_ext_perm} shows that expectation values of any externally implemented permutations  correspond to expectation values of permutations acting on internal degrees of freedom.   Therefore,  external symmetric weight reduces exactly to $F_{\mathrm{ind}}(\rho)$ and furthermore bounds from Ref.~\cite{englbrechtIndistinguishabilityIdenticalBosons2024} can be invoked directly as bounds on $F_{\mathrm{ind}}(\rho)$.  Their bounds depend on
\begin{equation}
P_2(\rho)=\frac{1+\tr{\Pi_2^{\text{ext}}\Omega_{\text{ext}}(\rho)}}{2}, 
\end{equation}
where $\Pi_2^{\text{ext}}$ is the normalized sum over all transpositions of the external modes,
\begin{equation}
    \Pi_{2}^{\text{ext}}=\frac{1}{\binom{N}{2}}
    \sum_{1\leq i<j\leq N}\Pi^{\text{ext}}_{(ij)}.
\end{equation}
Therefore, for one photon per input mode, $P_2(\rho)$ is obtained by averaging all pairwise Hong--Ou--Mandel visibilities. Estimating it therefore requires $O(N^2)$ two-photon interference experiments.

 In the present collision-free setting, the bounds from Ref.~\cite{englbrechtIndistinguishabilityIdenticalBosons2024} read
\begin{equation}\label{eq: P2-bound}
    \max\{0,(N-1)P_2(\rho)-(N-2)\}
    \leq\Find(\rho)\leq P_2(\rho).
\end{equation}
We compare these bounds with the upper bounds on $\Find(\rho)$ from  Theorem~\ref{th: Cyclic theorem} and improved lower bounds valid for product states,  and described in Appendix \ref{sec: Better bound}. Both bounds are accessible using samples from  \emph{a single}  Fourier interferometer (Protocol~\ref{Prot: Fourier}).

The comparison between both approaches for four common models of partial distinguishability is given in  Figure~\ref{fig:Comparison}. The $X$ model (panel~\subref{fig:comparison-x-model})~\cite{Tichy_2015} provides a controlled one-parameter interpolation in which the photons occupy pure internal states with uniform pairwise overlap $\langle\phi_i|\phi_j\rangle=x\in[0,1]$ for all $i\neq j$. The random time-delay model (panel~\subref{fig:comparison-time-delay}) describes the internal states as Gaussian temporal wave packets, with pairwise Hong--Ou--Mandel visibilities $|\langle t_j|t_k\rangle|^2=\exp[-(t_j-t_k)^2/(2\sigma^2)]$~\cite{tichy2014interference,shchesnovich2014partial,jonesInterferingDistinguishablePhotons2020a}. The results are averaged over random arrival times $t_j\sim\operatorname{Unif}(-1,1)$, leaving the wave-packet width $\sigma$ as the sole variable parameter. In the random-rotation model (panel~\subref{fig:comparison-rotations}), initially identical internal states undergo independent random unitary rotations, $|\phi\rangle\mapsto|\phi_j\rangle=e^{i\theta H_{\mathrm{rand},j}}|\phi\rangle$. For Fig.\ref{fig:Comparison}, we considered $\operatorname{dim}(\H_{\mathrm{int}})=d=4$; similar behavior is observed for all other dimensions $2\leq d\leq N$. As in the time-delay model, averaging over random realizations yields curves parameterized solely by the rotation strength $\theta$. Finally, the bad-batch model (panel~\subref{fig:comparison-bad-batch}) consists of a mixture of $N$ indistinguishable photons in the same internal state $|a_{0}\rangle$ and a state of fully distinguishable photons, i.e. $\rho=(1-\tau)|a_{0}\rangle\langle a_{0}|^{\otimes N} + \tau \otimes_{i=0}^{N-1} |a_{i}\rangle\langle a_{i}|$, with $\langle a_{i}|a_{j}\rangle=\delta_{i,j}$ and $\tau\in [0,1]$.

For all these families, the upper bound based on $P_c(\rho)$ is significantly tighter than the one obtained from pairwise HOM visibilities. Indeed, Appendix~\ref{sec: tightness} shows that, for independent photon sources, $P_c(\rho)\leq P_2(\rho)$, so the upper bound from Theorem~\ref{th: Cyclic theorem} is never weaker than the upper bound in Eq.~\eqref{eq: P2-bound}. The lower bound obtained from $P_c(\rho)$ is also stronger for the families analyzed here, although exceptions exist. Appendix~\ref{sec: tightness} gives an explicit example and provides sufficient conditions under which the lower bound of Theorem~\ref{th: Cyclic theorem} dominates. 

\section{Conclusion} \label{sec: conclusion}
Our work provides a scalable framework for certifying multiphoton indistinguishability. Using an isometric embedding of the internal photonic state into the bosonic Hilbert space, we identified the indistinguishability fidelity $F_{\mathrm{ind}}(\rho)$ with the expectation value of the projector onto the symmetric subspace of the internal degrees of freedom. We then presented two linear-optical protocols: a randomized protocol that directly estimates $F_{\mathrm{ind}}(\rho)$ for arbitrary internal states, and a single $N$-mode Fourier-interferometer protocol that yields tight two-sided bounds for sources preparing separable states by measuring the weight of the cyclic symmetric component of the state. In the high-indistinguishability regime, both protocols resolve a defect $\epsilon=1-F_{\mathrm{ind}}(\rho)$ using the optimal $\mathcal{O}(1/\epsilon)$ number of samples, independently of $N$. The cyclic bounds also provide a scalable alternative, showing significant improvements with respect to bounds based on average pairwise Hong--Ou--Mandel visibilities.

Several practical and theoretical questions remain open. First, our analysis assumes lossless propagation and ideal photon-number-resolving detection; incorporating non-unit efficiencies and dark counts \cite{kang2003dark} is important for applications to large photonic processors. Second, extending the framework beyond collision-free inputs to higher mode occupations, and to continuous-variable states such as those used in Gaussian boson sampling \cite{GaussianBS}, remains open. Finally, for general correlated states, it would be useful to reduce the shot-to-shot reconfiguration required by the randomized protocol, for example through unitary designs or partial derandomization, or to determine whether non-trivial bounds can be obtained from only a small number of fixed interferometers.

\section*{Acknowledgments }
 We thank Daniel Brod, Ra\'ul Garc\'ia-Patr\'on, Rawad Mezher, Ellen Derbyshire and Nathan Walk for interesting discussions. L.N. and E.F.G. acknowledge support from FCT-Fundação para a Ciência e a Tecnologia (Portugal) via the Project No. CEECINST/00062/2018 and from the project with the reference n.º 2023.15565.PEX, funded by national funds through FCT – Fundação para a Ciência e a Tecnologia, I.P.. M.R. is a FRIA grantee of the Fonds de la Recherche Scientifique – FNRS. E. F. G.  acknowledges funding from the National Council for Scientific and Technological Development – CNPq (Brazil) under grant 308292/2025-1. MO and LN acknowledges the support from the European Union’s Horizon Europe research and innovation program under EPIQUE Project GA No. 101135288. The C4QEC project is carried out within the IRAP of the Foundation for Polish Science co-financed by the European Union. N.J.C. acknowledges support from the Fonds de la Recherche Scientifique–FNRS (Belgium) under Grant No. T.0060.26 as well as project CHEQS within the Excellence of Science (EOS) program.

\section*{AI use disclosure}

The main ideas and results were conceived and developed by the authors. Chat-GPT (models 5.6 Sol and 6 Astra) were used to proofread, optimize and simplify results concerning bounds on $\Find(\rho)$ that use samples from a single Fourier interferometer (specifically technical results from Appendix \ref{sec: Cyclic projector bound},\ref{sec: Better bound} and \ref{sec: tightness}). Additionally, these models were used for polishing the narrative and spellchecking throughout the paper.
 
\bibliography{biblio}

\section{Appendix}\label{sec: appendix}
\appendix

The Appendix is organized as follows. Appendix~\ref{sec: bosemb} develops the bosonic embedding used in the main text and proves its relevant properties. Appendix~\ref{sec: Sample complexity} provides the concentration bounds for the cyclic and randomized estimators and establishes the optimal sample complexity in the near-perfect indistinguishability regime. Appendix~\ref{sec: Cyclic projector bound} proves the elementary fidelity bounds based on the cyclic projector. Appendix~\ref{sec: Better bound} derives the refined bounds using Bargmann invariants. Appendix~\ref{sec: tightness} compares the cyclic bounds with those obtained from average pairwise Hong--Ou--Mandel visibilities and analyzes when each is tighter. Lastly, Appendix~\ref{sec: iid estimation} presents a non-randomized estimation method for identical independent sources, Appendix~\ref{Sec: number of BI} quantifies the number of Bargmann invariants required to predict general multiphoton interference outcomes and Appendix~\ref{sec:negative_partition} discuss the case in which the partition twirling give rise to negative $c_{N}$ coefficient as discussed in the main text.

\section{Bosonic embedding and its properties}\label{sec: bosemb}
In this section we give a detailed description of the bosonic embedding and the fundamental property it satisfies. 

Let $\H_{in}$, $\H_{ext}$ be Hilbert space describing a internal and external degrees of freedom of a single particle Hilbert space. Internal degrees of freedom are not accessible experimentally (polarization, frequency etc.), while external (path) ones are accessible and measurable. A single particle Hilbert space is then $\H=\H_{in} \otimes \H_{ext}$. Let $\H_{tot}= \H^{\otimes N}$ be a Hilbert space of $N$ distinguishable particles. Upon relabeling of factors we have an isomorphism $\H_{tot}\approx \H_{int}^{\otimes N} \otimes \H_{ext}^{\otimes N}$. In $\H_{tot}$ we have a representation of the symmetric group $S_N$ by simultaneous permutation of both internal and external factors: ${\sigma} \in S_N\mapsto {\permrev{\hat{P}_{\pi}^{tot}}{\Pi_{\sigma}^{tot}}}={\permrev{\hat{P}_{\pi}^{int}}{\Pi_{\sigma}^{\mathrm{int}}}} \otimes {\permrev{\hat{P}_{\pi}^{ext}}{\Pi_{\sigma}^{\mathrm{ext}}}} \in \mathrm{U}(\H^{tot})$, where ${\permrev{\hat{P}_{\pi}^{int}}{\Pi_{\sigma}^{\mathrm{int}}}}$, ${\permrev{\hat{P}_{\pi}^{ext}}{\Pi_{\sigma}^{\mathrm{ext}}}}$ permute factors of $\H^{\otimes N}_{int}$ and $\H^{\otimes N}_{ext}$ respectively according to permutation ${\sigma}$. The Hilbert space of $N$ bosons with composite Hilbert space are then given by $\H^{bos}_{tot,N}= \Sym^N(\H_{int} \otimes \H_{ext}) \subset \H_{tot}$, which can be understood as image of $\H_{tot,N}$ under the projector $\Psym^{tot}=\frac{1}{N!}\sum_{{\sigma}\in S_N}{\permrev{\hat{P}_{\pi}^{tot}}{\Pi_{\sigma}^{tot}}}$.

\begin{restatable}[Bosonic Embedding]{Lemma}{Isometry}\label{lem:isometry}
Assume that $|\H_{ext}|>N$. Let $\ket{\mathbf{i}_0}=\ket{1}\otimes \ket{2} \otimes \ldots \otimes \ket{N}\in \H^{\otimes N}_{ext}$, with $\ket{i}$ being orthonormal vectors from $\H_{ext}$. Let $\ket{\psi}\in \H_{int}^{\otimes N}$  Let $V:\H_{int}^{\otimes N} \rightarrow \H^{bos}_{tot}$ be linear map defined by
\begin{equation}\label{eq:isoDEF}
    V \ket{\psi}\coloneq  \frac{1}{\sqrt{N!}} \sum_{{\sigma}\in S_N} {\permrev{\hat{P}_{\pi}^{tot}}{\Pi_{\sigma}^{tot}}} \left( \ket{\psi} \otimes \ket{\mathbf{i}_0}\right)\ .
\end{equation}
Then $V$ is an isometry, i.e. $V^\dagger V = \mathbb{I}_{\H^{\otimes N}_{int}}$, where $\mathbb{I}_{\H^{\otimes N}_{int}}$ denotes identity on $\H^{\otimes N}_{int}$.
\end{restatable}
\begin{proof}
    To prove that  $V$ is an isometry, it is enough to show that for arbitrary vectors $\ket{\psi},\ket{\phi}\in \H_{int}^{\otimes N}$ their images under $V$, given by  $\ket{\Psi}=V\ket{\psi}$ and $\ket{\Phi}=V \ket{\phi}$, satisfy $\bra{\Psi}\Phi\rangle =\bra{\psi}\phi\rangle$. To this end we compute
    \begin{equation}\label{eq:overlapISO}
        \bra{\Psi}\Phi\rangle = \frac{1}{N!} \sum_{{\sigma},\tau \in S_N} \bra{\psi}{\permrev{\hat{P}_{\pi^{-1} \tau}^{int}}{\Pi_{\sigma^{-1} \tau}^{\mathrm{int}}}}\ket{\phi} \bra{\mathbf{i}_0} {\permrev{\hat{P}_{\pi^{-1} \tau}^{ext}}{\Pi_{\sigma^{-1} \tau}^{\mathrm{ext}}}} \ket{\mathbf{i}_0}\ ,
    \end{equation}
    where we used the fact that ${\permrev{\hat{P}^{int/ext}}{\Pi^{\mathrm{int}/\mathrm{ext}}}}$ are unitary representations of $S_N$. We then observe that form the definition of ${\permrev{\hat{P}^{ext}}{\Pi^{\mathrm{ext}}}}$ and $\ket{\mathbf{i}_0}$ we have $\bra{\mathbf{i}_0} {\permrev{\hat{P}_{\pi^{-1} \tau}^{ext}}{\Pi_{\sigma^{-1} \tau}^{\mathrm{ext}}}} \ket{\mathbf{i}_0}=\delta_{{\sigma},\tau}$. Using this we see that the double sum  in Eq. \eqref{eq:overlapISO} simplifies to $\bra{\psi}\phi\rangle$. 
\end{proof}

\begin{rem}\label{Rem: Isometry property}
     Let $\rho, \sigma$ be two states supported on $\H_{int}^{\otimes N}$. Define $\Omega(\rho) \coloneq V \rho V^{\dagger}$. Because $V$ is an isometry we have
     \begin{equation}
         F\left(\rho,\sigma\right)= F\left(\Omega(\rho),\Omega(\sigma)\right)\ ,
     \end{equation}
     and an analogous relation holds for the trace distance between states. Because of this,  closeness of states in $\H_{int}^{\otimes N}$ is equivalent to closeness of their "photonic representations" $\Omega(\rho),\Omega(\sigma)\in \H_{tot}^{bos}$. 

\end{rem}

Additionally, as discussed in the main text, it is possible to use the isometry property to prove the following lemmas.

\Fidtoclose*

\begin{proof}
    Since $V$ is an isometry, we have that $F\left(\Omega(\rho),\Omega(\sigma)\right)=F(\rho,\sigma)$, which implies that
    \begin{align*}
        \Find(\rho)&=\max_{\sigma \in \D\left(\Sym^N(\H_{int})\right)} F\left(\rho,\sigma\right)\\
        &=\max_{\sigma \in \D\left(\Sym^N(\H_{int})\right)}\lVert\sqrt{\rho}\sqrt{\sigma} \rVert_{1}^{2}\\
        &=\tr{\Psym^{\mathrm{int}}\rho}\max_{\sigma}\lVert\sqrt{\tilde{\rho}}\sqrt{\sigma} \rVert_{1}^{2}=\tr{\Psym^{\mathrm{int}}\rho}  \ .
    \end{align*}
\end{proof}

\LemDist*

\begin{proof}
In the first-quantization representation on $\mathcal{H}_{ext}^{\otimes N}$, the state $\ket{\mathbf{n}_0}$ corresponding to $N$ indistinguishable photons is given by $\ket{\mathbf{n}_0} = \frac{1}{\sqrt{N!}}\sum_{{\sigma}\in S_N}{\permrev{\hat{P}_\pi^{ext}}{\Pi_\sigma^{\mathrm{ext}}}}\ket{\mathbf{i}_0} = \sqrt{N!}\,\Psym^{\mathrm{ext}}\ket{\mathbf{i}_0}$. On the support of $\Omega_{ext}(\rho)$, which lies in $\mathrm{span}\{{\permrev{\hat{P}_\pi^{ext}}{\Pi_\sigma^{\mathrm{ext}}}}\ket{\mathbf{i}_0}\}_{{\sigma}\in S_N}$, the symmetric projector acts as the rank-$1$ operator $\Psym^{\mathrm{ext}} = \kb{\mathbf{n}_0}{\mathbf{n}_0}$.

Multiplying the expression for $\Omega_{ext}(\rho)$ in Eq.~\eqref{eq:reducedstate} by $\Psym^{\mathrm{ext}}$ from the left, and using the group property $\Psym^{\mathrm{ext}}{\permrev{\hat{P}_\tau^{ext}}{\Pi_\tau^{\mathrm{ext}}}} = \Psym^{\mathrm{ext}}$, we obtain
\begin{align}
&\Psym^{\mathrm{ext}}\Omega_{ext}(\rho) = \frac{1}{N!}\sum_{\tau,\tau'\in S_N} \Psym^{\mathrm{ext}}\ket{\mathbf{i}_0}\bra{\mathbf{i}_0}{\permrev{\hat{P}_{\tau'}^{ext}}{\Pi_{\tau'}^{\mathrm{ext}}}} \tr{\rho {\permrev{\hat{P}_{\tau'\tau^{-1}}^{int}}{\Pi_{\tau'\tau^{-1}}^{\mathrm{int}}}}} \\
&= \Psym^{\mathrm{ext}}\ket{\mathbf{i}_0}\left(\sum_{\tau\in S_N}\bra{\mathbf{i}_0}{\permrev{\hat{P}_\tau^{ext}}{\Pi_\tau^{\mathrm{ext}}}}\right) \tr{\rho \left(\frac{1}{N!}\sum_{\sigma\in S_N}{\permrev{\hat{P}_\sigma^{int}}{\Pi_\sigma^{\mathrm{int}}}}\right)} \\
&= \left(\frac{1}{\sqrt{N!}}\ket{\mathbf{n}_0}\right)\left(\sqrt{N!}\bra{\mathbf{n}_0}\right) \tr{\rho \Psym^{\mathrm{int}}} \\
&= \Find(\rho)\kb{\mathbf{n}_0}{\mathbf{n}_0},
\end{align}
where we substituted $\sigma = \tau'\tau^{-1}$ and used Lemma~\ref{lem:closestIND}. Since $\Psym^{\mathrm{ext}} = \kb{\mathbf{n}_0}{\mathbf{n}_0}$, this establishes that $\kb{\mathbf{n}_0}{\mathbf{n}_0}\Omega_{ext}(\rho) = \Find(\rho)\kb{\mathbf{n}_0}{\mathbf{n}_0}$, proving that $\ket{\mathbf{n}_0}$ is an eigenvector of $\Omega_{ext}(\rho)$ with eigenvalue $\Find(\rho)$.

Decomposing $\Omega_{ext}(\rho)$ onto $\mathrm{span}\{\ket{\mathbf{n}_0}\}$ and its orthogonal complement yields no cross terms:
\begin{equation}
\Omega_{ext}(\rho) = \Find(\rho)\kb{\mathbf{n}_0}{\mathbf{n}_0} + (1-\Find(\rho))\sigma^\perp,
\end{equation}
where $\sigma^\perp = \frac{(\mathbb{I}-\kb{\mathbf{n}_0}{\mathbf{n}_0})\Omega_{ext}(\rho)(\mathbb{I}-\kb{\mathbf{n}_0}{\mathbf{n}_0})}{1-\Find(\rho)}$ is a density matrix orthogonal to $\kb{\mathbf{n}_0}{\mathbf{n}_0}$. 

Finally, because $\kb{\mathbf{n}_0}{\mathbf{n}_0}$ and $\sigma^\perp$ have mutually orthogonal supports, the trace distance evaluates to
\begin{align*}
&\dtr(\Omega_{ext}(\rho),\kb{\mathbf{n}_0}{\mathbf{n}_0}) = \frac{1}{2}\|\Omega_{ext}(\rho)-\kb{\mathbf{n}_0}{\mathbf{n}_0}\|_1 \\
&= \frac{1-\Find(\rho)}{2}\left(\|\kb{\mathbf{n}_0}{\mathbf{n}_0}\|_1 + \|\sigma^\perp\|_1\right) = 1-\Find(\rho),
\end{align*}
which completes the proof.
\end{proof}

We study now the behavior of permutation in relation to the bosonic embedding. Consider a unitary operation ${\permrev{\hat{P}_{\pi}^{ext}}{\Pi_{\sigma}^{\mathrm{ext}}}}\in\mathrm{U}(\H_{ext})$ that permutes the first $N$ basis vectors (modes) in $\H_{ext}$, and leaves the remaining ones intact: ${\permrev{\hat{P}_{\pi}^{ext}}{\Pi_{\sigma}^{\mathrm{ext}}}} \ket{i}= \ket{{\sigma}(i)}$ for $i\in[N]$ and ${\permrev{\hat{P}_{\pi}^{ext}}{\Pi_{\sigma}^{\mathrm{ext}}}}\ket{i} =\ket{i}$ for $i>N$. The following lemma shows that when ${\permrev{\hat{P}_\pi^{ext}}{\Pi_\sigma^{\mathrm{ext}}}}$ (despite acting on \emph{external} degrees) of freedom effectively implements a permutation of internal degrees of freedom for states of the form $V\ket{\psi}$ 

\begin{restatable}[External permutation of modes implements permutation of internal states]{Lemma}{ExtPermIntPerm}\label{lemma: ext to int}
    Let $\ket{\psi}\in\H_{int}^{\otimes N}$ and $V$ be defined in Eq.\eqref{eq:isoDEF}. Let ${\permrev{\hat{P}_{\pi}^{ext}}{\Pi_{\sigma}^{\mathrm{ext}}}}\in \mathrm{U}(\H_{ext})$ be the mode permutation defined above and let $\mathbb{I}^{\mathrm{int}}$ denote the identity operator on $\H_{int}^{\otimes N}$. Then we have
\begin{equation}\label{eq:externalINTWERNALpermutation}
        (\mathbb{I}^{\mathrm{int}}\otimes {\permrev{\hat{P}_{\pi}^{ext}}{\Pi_{\sigma}^{\mathrm{ext}}}}) V \ket{\psi} =  V \left({\permrev{\hat{P}_{\pi^{-1}}^{int}}{\Pi_{\sigma^{-1}}^{\mathrm{int}}}} \ket{\psi}\right)\ .
    \end{equation}
    \end{restatable}
\begin{proof}
We expand $V\ket{\psi}$ using \eqref{eq:isoDEF}:

\begin{gather}
        (\mathbb{I}^{\mathrm{int}}\otimes {\permrev{\hat{P}_{\pi}^{ext}}{\Pi_{\sigma}^{\mathrm{ext}}}}) V \ket{\psi} = \sum_{\tau\in S_N} \frac{{\permrev{\hat{P}_{\tau}^{int}}{\Pi_{\tau}^{\mathrm{int}}}} \otimes  {\permrev{\hat{P}_{\pi}^{ext}}{\Pi_{\sigma}^{\mathrm{ext}}}} {\permrev{\hat{P}_{\tau}^{ext}}{\Pi_{\tau}^{\mathrm{ext}}}}}{\sqrt{n!}}  \left( \ket{\psi} \otimes \ket{\mathbf{i}_0}\right) .
\end{gather}
We relabel ${\sigma} \tau={\permrev{\sigma}{\nu}} $ which implies $\tau={\sigma}^{-1}{\permrev{\sigma}{\nu}}$, and the above can be rewritten as
\begin{align*}
        (\mathbb{I}^{\mathrm{int}}\otimes {\permrev{\hat{P}_{\pi}^{ext}}{\Pi_{\sigma}^{\mathrm{ext}}}}) V \ket{\psi} &=    \sum_{{\permrev{\sigma}{\nu}}\in S_N} \frac{{\permrev{\hat{P}_{\pi^{-1}\sigma}^{int}}{\Pi_{\sigma^{-1}\nu}^{\mathrm{int}}}} \otimes  {\permrev{\hat{P}_{\sigma}^{ext}}{\Pi_{\nu}^{\mathrm{ext}}}}}{\sqrt{n!}} \left( \ket{\psi} \otimes \ket{\mathbf{i}_0}\right) \\
        &=( {\permrev{\hat{P}_{\pi^{-1}}^{int}}{\Pi_{\sigma^{-1}}^{\mathrm{int}}}}\otimes\mathbb{I}^{\mathrm{ext}}) V \ket{\psi} 
\end{align*}

\end{proof}
In other words, a permutation  of external modes ${\permrev{\hat{P}_{\pi}^{ext}}{\Pi_{\sigma}^{\mathrm{ext}}}}$ applied to arbitrary state on the image of $V$ (i.e. a state of partially distinguishable photons, each localized in exactly one external mode) amounts to implementing the inverse permutation on the internal degrees of freedom of the particles. This result was proven with a different formalism in \cite{Geller2025,schadow2026certificationlinearopticalquantum}.

\begin{restatable}[External modes projectors implement internal modes projectors]{Corollary}{External_mode_permutation}\label{cor: Projection equivalence}
    Given the projectors $\Pc^{\mathrm{ext}}$ and $\Psym^{\mathrm{ext}}$ we have
    \begin{align}
        \tr{\Pc^{\mathrm{ext}} V|\psi\rangle\langle\psi|V^{\dagger}}&=\tr{\Pc^{\mathrm{int}}|\psi\rangle\langle\psi|}\\
        \tr{\Psym^{\mathrm{ext}} V|\psi\rangle\langle\psi|V^{\dagger}}&=\tr{\Psym^{\mathrm{int}}|\psi\rangle\langle\psi|}
    \end{align}
\end{restatable}
\begin{proof}
    We give the proof for $\Psym^{\mathrm{ext}}$, but the same can be done for the cyclic one. From the definition $\Psym^{\mathrm{ext}}=\frac{1}{N!}\sum_{{\sigma} \in S_{N}}{\permrev{\hat{P}_{\pi}^{ext}}{\Pi_{\sigma}^{\mathrm{ext}}}}\otimes \mathbb{I}^{\mathrm{int}}$. We make use of lemma \ref{lemma: ext to int} we have
    \begin{align*}
        \tr{\Psym^{\mathrm{ext}}V|\psi\rangle\langle\psi|V^{\dagger}} &= \frac{1}{N!}\sum_{{\sigma} \in S_{N}}\tr{V{\permrev{\hat{P}_{\pi^{-1}}^{int}}{\Pi_{\sigma^{-1}}^{\mathrm{int}}}}|\psi\rangle\langle\psi|V^{\dagger}}\\
        &=\frac{1}{N!}\sum_{{\sigma} \in S_{N}}\tr{{\permrev{\hat{P}_{\pi}^{int}}{\Pi_{\sigma}^{\mathrm{int}}}}|\psi\rangle\langle\psi|}\\
        &=\tr{\Psym^{\mathrm{int}}|\psi\rangle\langle\psi|}
    \end{align*}
    Notice that in the second line we made use of lemma \ref{lem:isometry} when we simplified $V^\dagger V = \mathbb{I}_{\H^{\otimes N}_{int}}$ by use of the trace property.
\end{proof}
The corollary just states that if we are able to perform a projector in the external mode space, we are able to perform a projector in the internal space. Notice that with the approach discussed in the main text we can do the cyclic one directly, but not the symmetric one, which requires randomization.

\section{Sample complexity of certification of indistinguishability}\label{sec: Sample complexity}

This appendix present details relevant for sample complexity bounds of Protocols \ref{Prot: Fourier} and \ref{Prot: Randomized}. Section \ref{sec: estimator concentration} collects the relevant concentration inequalities while  Section~\ref{sec: protocols} and establishes their optimality in the near-perfect regime.

\subsection{Concentration bounds for the two estimators}\label{sec: estimator concentration}

\begin{Lemma}[Bernstein's inequality~{\cite{boucheron2013concentration}}]\label{lem:bernstein_concentration}
Let $Z_1,\ldots,Z_r$ be independent random variables with common mean $\mu$. Suppose that $|Z_m-\mu|\leq b$ almost surely and $\operatorname{Var}(Z_m)\leq v$ for every $m$, where $b>0$ and $v\geq0$. Then their empirical mean $\overline Z=r^{-1}\sum_{m=1}^r Z_m$ satisfies
\begin{equation}\label{eq:appendix_bernstein_general}
    \Pr\!\left(|\overline Z-\mu|\geq\eta\right)
    \leq 2\exp\!\left[-\frac{r\eta^2}{2v+\frac{2}{3}b\eta}\right],
    \qquad \eta>0.
\end{equation}
Consequently, for any $\delta\in(0,1)$, the estimate $\overline Z$ has additive error less than $\eta$ with probability at least $1-\delta$ whenever
\begin{equation}\label{eq:appendix_bernstein_samples}
    r\geq\left(\frac{2v}{\eta^2}+\frac{2b}{3\eta}\right)
    \ln\!\frac{2}{\delta}.
\end{equation}
\end{Lemma}

The sample complexity bounds for both protocols follow straightforwardly. For the cyclic estimator in Protocol~\ref{Prot: Fourier}, take $Z_m=X_m$, $\mu=P_c(\rho)$, $v=P_c(\rho)[1-P_c(\rho)]$ and $b=1$. Substitution gives Eq.~\eqref{eq:cyclic_sample_bound_main} directly. For the randomized estimator in Protocol~\ref{Prot: Randomized}, take $Z_m=Y_m$, $\mu=\Find(\rho)$, $v=1-\Find^2(\rho)$ and $b=2$, giving Eq.~\eqref{eq:randomized_sample_bound_main}.

\subsection{Optimality in the near-perfect regime}\label{sec: sample optimality}

We now show that no certification procedure can improve the
\(\mathcal{O}(\varepsilon^{-1}\log\delta^{-1})\) scaling in general, even if it may perform an arbitrary collective measurement on all copies. Let \(\ket{0},\ket{1}\in\H_{\mathrm{int}}\) be orthonormal and consider the pure product states
\begin{align}
    \ket{\Phi_0}
    &=\ket{0}^{\otimes N},
    \label{eq:optimality_perfect_product_state}\\
    \ket{\Phi_\alpha}
    &=
    \left(
    \sqrt{1-\alpha}\ket{0}
    +\sqrt{\alpha}\ket{1}
    \right)
    \otimes\ket{0}^{\otimes(N-1)},
    \qquad 0<\alpha<1.
    \label{eq:optimality_defective_product_state}
\end{align}
The first state is supported on symmetric subspace and therefore describes perfectly indistinguishable photons. A straightforward computation gives,
\begin{equation}\label{eq:optimality_defect_relation}
    \Find\!\left(\kb{\Phi_\alpha}{\Phi_\alpha}\right)
    =P_c\!\left(\kb{\Phi_\alpha}{\Phi_\alpha}\right)
    =1-\frac{N-1}{N}\alpha
    \coloneq1-\varepsilon.
\end{equation}
The two states have squared overlap
\begin{equation}\label{eq:optimality_product_overlap}
    \left|\langle\Phi_0|\Phi_\alpha\rangle\right|^2
    =1-\alpha
    =1-\frac{N}{N-1}\varepsilon.
\end{equation}
Since the bosonic embedding is an isometry, the same overlap holds for the corresponding physical photonic states.

The minimum equal-prior error probability for discriminating \(r\) copies of two pure states is the Helstrom error~\cite{helstrom1976quantum}. In the present case it is
\begin{equation}\label{eq:optimality_helstrom_error}
    P_{\mathrm{err}}^{(r)}
    =\frac{1}{2}
    \left[
    1-\sqrt{1-(1-\alpha)^r}
    \right].
\end{equation}
Any procedure that certifies a defect of size \(\varepsilon\) with error probability at most \(\delta<1/2\) must, in particular, distinguish the two states above with
\(P_{\mathrm{err}}^{(r)}\leq\delta\).
Equation~\eqref{eq:optimality_helstrom_error} then implies
\begin{equation}\label{eq:optimality_copy_lower_bound}
    r
    \geq
    \frac{
    \ln\!\left[1/\!\left(4\delta(1-\delta)\right)\right]
    }{
    -\ln(1-\alpha)
    }
    =
    \Omega\!\left(
    \frac{1}{\varepsilon}
    \ln\!\frac{1}{\delta}
    \right).
\end{equation}
Here we used
\(\alpha=N\varepsilon/(N-1)=\mathcal{O}(\varepsilon)\)
and
\(-\ln(1-\alpha)=\mathcal{O}(\alpha)\)
as \(\varepsilon\rightarrow0\). This lower bound already holds within the pure-product source model relevant to the cyclic protocol. Together with the upper bounds obtained by substituting the near-perfect variance estimates into Eqs.~\eqref{eq:randomized_sample_bound_main} and~\eqref{eq:cyclic_sample_bound_main}, it proves the optimality of the near-perfect sample-complexity scaling for both protocols.

\section{Cyclic projector bound}\label{sec: Cyclic projector bound}
We start this section by proving first a bound for separable pure states in terms of the expectation value of the cyclic projector $\Pc$. We will then extend the result to mixed states. Lastly, we will strengthen the bound by using, in addition to the cycle projector expectation value, the values of the single Bargmann invariants.

We start by defining the Bargmann multiplicative score of the state $\rho$ as
\begin{equation}\label{eq: Bargmann multiplicative score}
    M(\rho)=\prod_{k=1}^{N-1}\left|\tr{{\permrev{\hat{P}_{C^k}^{(int)}}{\Pi_{C^k}^{\mathrm{int}}}} \rho}\right|^{\frac{2}{N}} .
\end{equation}

We have the following Lemma.
\begin{restatable}[Bargmann multiplicative score bound]{Lemma}{Bound pc}\label{lem: bound pc}
Let $P_{c}(\rho)\geq 1/2$ then we have
\begin{equation}
    1\geq |\tr{{\permrev{\hat{P}_{C^k}^{(int)}}{\Pi_{C^k}^{\mathrm{int}}}} \rho}|\geq 2P_{c}(\rho)-1
\end{equation} 
which implies 
\begin{equation}
    M(\rho)\geq 2P_{c}(\rho)-1.
\end{equation}
\end{restatable}
\begin{proof}
    Let us start by defining $x_k=|\tr{{\permrev{\hat{P}_{C^k}^{(int)}}{\Pi_{C^k}^{\mathrm{int}}}} \rho}|$. The upper bound on the $x_k$ is trivial. For the lower bound, we can use the results of \cite{novo2026nativelinearopticalprotocolefficient}, to rewrite
\begin{equation}
    \tr{{\permrev{\hat{P}_{C^k}^{(int)}}{\Pi_{C^k}^{\mathrm{int}}}} \rho}=P_{c}(\rho)+\sum_{r=1}^{N-1}p_r \omega^{rk}  
\end{equation}
where $p_r$ are the expectation value of the projectors orthogonal to the cyclic one. By using the reverse triangular inequality we have
\begin{align}
    x_k&=|\tr{{\permrev{\hat{P}_{C^k}^{(int)}}{\Pi_{C^k}^{\mathrm{int}}}} \rho}|\geq P_{c}(\rho)-\sum_{r=1}^{N-1}p_r\\
    & = P_{c}(\rho) - (1-P_{c}(\rho))=2P_{c}(\rho) - 1 .
\end{align}

To bound the Bargmann multiplicative score, we can use the property of the Fourier transform, to obtain 
\begin{equation}
    \sum_{k=1}^{N-1}\Re\left\{\tr{{\permrev{\hat{P}_{C^k}^{(int)}}{\Pi_{C^k}^{\mathrm{int}}}} \rho}\right\}=NP_{c}(\rho)-1.
\end{equation}

Since $x_k$ is the modulus of the above quantities we can rewrite it as 
\begin{align}\label{eq: partial result}
    \sum_{k=1}^{N-1}(1-x_k)\leq N(1-P_{c}(\rho) )= \frac{N}{2}(1-\left[2P_{c}(\rho)-1\right]).
\end{align}

We can the use the following inequality 
\begin{equation}
    \log(x)\geq \frac{1-x}{1-q}\log(q)  \ \ \forall q\in (0,1] .
\end{equation}
Applying the above for all $k$ in the previous equation and substituting $q=2P_{c}(\rho)-1$, we have

\begin{equation}
    \sum_{k=1}^{N-1}\log(x_k)\geq \frac{N}{2}\log(2P_{c}(\rho)-1)
\end{equation}
which once we exponentiate and raise both side to the power $2/N$ gives us the wanted bound.
\end{proof}

We can now use the above to bound the fidelity for pure states. 
\begin{restatable}[Pure state bound]{Lemma}{Bound pure}\label{lem: bound pc}
Let $|\Phi\rangle=|\phi_1\rangle\otimes\dots\otimes|\phi_{N}\rangle$ then we have
\begin{equation}
    \Find(|\Phi\rangle\langle\Phi|)\geq M(|\Phi\rangle\langle\Phi|).
\end{equation}

\end{restatable}
\begin{proof}
    For each $i$. the normalized vector $|\phi_{i}\rangle$ belongs to the fully symmetrized subspace, hence the fidelity can be bounded by
    \begin{equation}
         \Find(|\Phi\rangle\langle\Phi|)\geq R_{i}=\prod_{j=1}^{N}|\langle \phi_{i}|\phi_{j}\rangle|^{2} \ \ \forall i .
    \end{equation}
    We can rewrite it as 
    \begin{align}
        \Find(|\Phi\rangle\langle\Phi|)\geq \max_{i}R_{i}\geq \frac{\sum_{i}R_{i}}{N}\geq \left(\prod_{i}R_{i}\right)^{\frac{1}{N}}
\end{align}
where in the last step we made use of the arithmetic-geometric mean inequality. Notice now that we can rearrange the right hand side to obtain the Bargmann multiplicative score
\begin{equation}
    \prod_{i}R_{i}^{1/N}=\prod_{k=1}^{N-1}\prod_{i=1}^{N}|\langle\phi_{i}|\phi_{j}\rangle|^{2/N}=M(|\Phi\rangle\langle\Phi|)
\end{equation}
which concludes the proof.
\end{proof}

Lastly, we can generalize to the case of generic separable states, and thus give a proof of the Theorem \ref{th: Cyclic theorem}. 

\begin{proof}
    The existence of the protocol is provided in the main text. The upper bound can be easily derived by noticing that $\Pc-\Psym\succeq 0$, and thus $P_{c}(\rho)\geq \Find(\rho) \ \forall \rho$. The lower bound can be derived from Lemma \ref{lem: bound pc}, simply by realizing that the function 
    \begin{equation}
        g(p)=\max\left\{0,2p-1\right\}
    \end{equation}
    is a convex function. As a consequence we have that if we look into the pure state decomposition of $\rho$, we have $\rho=\sum_{a}w_{a}|\Phi_{a}\rangle\langle \Phi_{a}|$ for some probability distribution $\boldsymbol{w}$. Then we have
      \begin{align}
        \Find(\rho)
        &\geq\sum_a w_a g\!\left(P_c(\Phi_a)\right)
        \notag\\
        &\geq g\!\left(\sum_a w_a P_c(\Phi_a)\right)
        =g\!\left(P_c(\rho)\right),
        \label{eq:appendix_separable_cyclic_lower_bound}
    \end{align}
    where in the last step we made use of the Jensen inequality.
\end{proof}

\section{A better bound beyond the cyclic projection}\label{sec: Better bound}
In this section, we derive stronger bounds on $\Find(\rho)$ for product internal states $\rho=\bigotimes_{i=1}^N\rho_i$ using Bargmann invariants. The nonlinear refinements do not extend to arbitrary separable mixtures, as illustrated in Remark~\ref{rem: refined bound product assumption}. It is important to clarify that this bound is not experimentally harder, and actually as discussed in \cite{novo2026nativelinearopticalprotocolefficient}, different post-process of the data from a Fourier interferometer provides the expectation value of the Bargmann invariants. To do so, we invoke a similar approach to the one proposed in \cite{englbrechtIndistinguishabilityIdenticalBosons2024} with the use of the operator $\Pi_{2}$. In the following, we will implicitly consider the case $N\geq 4$, whereas the case $N\leq 4$ are trivial and better bound can be provided. We introduce the Bargmann additive score as
\begin{equation}
    A(\rho)=\sum_{k=1}^{N-1}|\tr{{\permrev{\hat{P}_{C^k}^{(int)}}{\Pi_{C^k}^{\mathrm{int}}}} \rho}|^{2/N} 
\end{equation}
which can be derived with the same setup necessary for measuring $P_{c}(\rho)$ as shown in \cite{novo2026nativelinearopticalprotocolefficient}. The optimal bounds can then be written as
\begin{align}\label{eq: Final bounds}
    \Find(\rho)&\geq \frac{NA(\rho)}{4(N-1)}+\frac{N-2}{2(N-1)}P_{c}(\rho)-\frac{N(N-3)}{4(N-1)} \\
    &\geq \frac{N(N(2P_{c}(\rho)-1)^{\frac{2}{N}} - (N-4))+4(N-2)P_{c}(\rho)}{8(N-1)}
\end{align}
where the first inequality in term of the Bargmann additive score is stricter than the second. Although the second inequality is looser, it has the advantage to depend on a single observable, namely $P_{c}(\rho)$. We recall that this lower bounds must be always compared with the one provided in Theorem~\ref{th: Cyclic theorem}.

The derivation is organized as follows. Section~\ref{sec: notation} introduces Young diagrams and the irreducible representations of the symmetric group. Section~\ref{sec: homothety} then shows that $\Pi_2$ acts as a scalar multiple of the identity on each irreducible representation. Sections~\ref{sec: largest eig} and~\ref{sec: largest eig cycle} determine its largest eigenvalues outside the symmetric subspace and within the cyclic-invariant subspace, respectively. These results are combined in Section~\ref{sec: operator bound} to establish an operator inequality, derive the fidelity bounds in Eq.~\eqref{eq: Final bounds}, and analyze their behavior near perfect indistinguishability. Section~\ref{sec: cyclic bound improvement} identifies the regime in which the refined lower bound improves upon the elementary bound of Theorem~\ref{th: Cyclic theorem}. Section~\ref{sec: refined cyclic confidence} extends this refinement to finite data, yielding a confidence interval for $\Find(\rho)$ and clarifying the necessity of the product-state assumption. Finally, Section~\ref{sec:optimality} establishes the optimality of the coefficients in the operator inequality involving $\Pi_2$ and $\Pc$.

\subsection{Group theory notation}\label{sec: notation}
Let $n$ be a positive integer. A partition of $n$, denoted $\lambda \vdash n$, is a finite sequence of weakly decreasing positive integers $\lambda = (\lambda_{1}, \lambda_{2}, \dots, \lambda_{k})$ such that $\sum_{i=1}^{k}\lambda_{i}=n$. The integers $\lambda_{i}$ are the \emph{parts} of the partition, and $k$ is its length. Partitions are visually represented by \emph{Young diagrams} (or Ferrers diagrams). A Young diagram is a collection of boxes arranged in left-justified rows, where the $i$-th row from the top contains exactly $\lambda_{i}$ boxes. The \emph{transpose} (or conjugate) of a partition, denoted $\overline{\lambda}$, is obtained by reflecting the Young diagram of $\lambda$ across its main diagonal, swapping rows and columns. For example, given $n=8$ particles and the partition $\lambda=(4,3,1)$, we have:
\begin{equation}
    \lambda =\ydiagram{4,3,1} \implies \overline{\lambda} = \ydiagram{3,2,2,1} \quad .
\end{equation}

We say that the partition $\lambda$ majorizes (or dominates) the partition $\mu$, and write it as $\lambda \succ \mu$, if
\begin{equation}
    \sum_{i=1}^{k}\lambda_{i}\geq \sum_{i=1}^{k}\mu_{i} \ \ \forall 1\leq k\leq n \ .
\end{equation}

We decompose the $S_N$-representation space into irreducible components as
\begin{equation}
    \mathcal{H}_{\text{ext}}^{\otimes N} \simeq \bigoplus_{\lambda \vdash N} V_\lambda \otimes \mathcal{M}_\lambda,
\end{equation}
where $V_\lambda$ denotes the irreducible representation corresponding to the partition $\lambda \vdash N$, and $\mathcal{M}_\lambda$ is the associated multiplicity space.

\subsection{\texorpdfstring{$\Pi_{2}$}{Pi2} is a homothety}\label{sec: homothety}

We start by recalling the definition of $\Pi_{2}$ 
\begin{equation}
    \Pi_{2}=\frac{1}{\binom{N}{2}}\sum_{1\leq i<j\leq N} {\permrev{\hat{P}_{(ij)}}{\Pi_{(ij)}}}
\end{equation}
where $P_{(ij)}$ is the permutation associated with the transposition of $i$ into $j$. We want to show that 
\begin{equation}
    \Pi_{2}{\bigg |}_{V_\lambda \otimes \mathcal{M}_\lambda}=t_{\lambda}\Pi_{\lambda}
\end{equation}
with $\Pi_{\lambda}$ the projector over the irreducible representation of the symmetric group, which acts as the identity operator in the $\lambda$ subspace. We recall the definition of such projector
\begin{equation}
    \Pi_{\lambda}=\frac{\chi_{\lambda}(e)}{N!}\sum_{{\sigma} \in S_{N}}\chi_{\lambda}(\sigma){\permrev{\hat{P}_{\pi}}{\Pi_{\sigma}}}
\end{equation}
where $\chi_{\lambda}({\sigma})$ is the character of the permutation ${\sigma}$ in the irreducible representation sector $\lambda$. We start by considering the following lemma.
\begin{restatable}[]{Lemma}{permutation projection}\label{lemma: permutation projection}
${\permrev{\hat{P}_{\pi}}{\Pi_{\sigma}}}\Pi_{\lambda}=\Pi_{\lambda}{\permrev{\hat{P}_{\pi}}{\Pi_{\sigma}}}$ .

\end{restatable}
\begin{proof}
    We can directly compute it as 
    \begin{align}
   {\permrev{\hat{P}_{\pi}}{\Pi_{\sigma}}}\Pi_{\lambda}&=\frac{\chi_{\lambda}(e)}{N!}\sum_{\tau}\chi_{\lambda}(\tau){\permrev{\hat{P}_{\pi}}{\Pi_{\sigma}}}{\permrev{\hat{P}_{\tau}}{\Pi_{\tau}}}\\
        &= \frac{\chi_{\lambda}(e)}{N!}\sum_{\tau}\chi_{\lambda}(\tau){\permrev{\hat{P}_{\pi}}{\Pi_{\sigma}}}{\permrev{\hat{P}_{\tau}}{\Pi_{\tau}}}{\permrev{\hat{P}_{\pi^{-1}}}{\Pi_{\sigma^{-1}}}}{\permrev{\hat{P}_{\pi}}{\Pi_{\sigma}}}\\
        &= \frac{\chi_{\lambda}(e)}{N!}\sum_{\tau}\chi_{\lambda}({\sigma}\tau{\sigma}^{-1}){\permrev{\hat{P}_{\pi}}{\Pi_{\sigma}}}{\permrev{\hat{P}_{\tau}}{\Pi_{\tau}}}{\permrev{\hat{P}_{\pi^{-1}}}{\Pi_{\sigma^{-1}}}}{\permrev{\hat{P}_{\pi}}{\Pi_{\sigma}}}\\
        &=\frac{\chi_{\lambda}(e)}{N!}\sum_{{\permrev{\sigma}{\nu}}}\chi_{\lambda}({\permrev{\sigma}{\nu}}){\permrev{\hat{P}_{\sigma}}{\Pi_{\nu}}}{\permrev{\hat{P}_{\pi}}{\Pi_{\sigma}}}\\
        &=\Pi_{\lambda}{\permrev{\hat{P}_{\pi}}{\Pi_{\sigma}}} .
    \end{align}
    Where we have used the fact that $\chi_{\lambda}({\sigma}\tau{\sigma}^{-1})=\chi_{\lambda}(\tau)$ since it is a class function, and thus invariant under conjugation.
\end{proof}
Now notice that, by the means of Lemma \ref{lemma: permutation projection} we have that 
\begin{equation}
\Pi_{\lambda}{\permrev{\hat{P}_{\pi}}{\Pi_{\sigma}}}\Pi_{\lambda}=\Pi_{\lambda}{\permrev{\hat{P}_{\pi}}{\Pi_{\sigma}}}
\end{equation}
which implies that $\Pi_{\lambda}{\permrev{\hat{P}_{\pi}}{\Pi_{\sigma}}}$ is the representation of the permutation ${\sigma}$ in the irreducible representation associated with the $\lambda$ partition. 

We can also notice that the following.
\begin{restatable}[]{Lemma}{Pi2}\label{cor: Pi2}
    ${\permrev{\hat{P}_{\pi}}{\Pi_{\sigma}}}\Pi_{2}{\permrev{\hat{P}_{\pi}^{\dagger}}{\Pi_{\sigma}^{\dagger}}}=\Pi_{2}.$
\end{restatable}
\begin{proof}
    To prove it we can simply notice that for every transposition ${\permrev{\hat{P}_{(ij)}}{\Pi_{(ij)}}}$ we have 
    \begin{equation}
        {\permrev{\hat{P}_{\pi}}{\Pi_{\sigma}}}{\permrev{\hat{P}_{(ij)}}{\Pi_{(ij)}}}{\permrev{\hat{P}_{\pi}^{\dagger}}{\Pi_{\sigma}^{\dagger}}}={\permrev{\hat{P}_{(\pi(i),\pi(j))}}{\Pi_{(\sigma(i),\sigma(j))}}}
    \end{equation}
    which is another transposition. If we sum over all the ordered pairs and normalize, we recover $\Pi_{2}$.
\end{proof}

We are going to use the above to apply one corollary of the Schur's Lemma, for details see \cite{serreLinearRepresentationsFinite1977} (Corollary 2.16.1).

\begin{restatable}[]{Corollary}{Schur}\label{cor: Schur}
Let $\alpha : G \to \mathrm{GL}(V_{\alpha})$ and $\beta : G \to \mathrm{GL}(V_{\beta})$ be two irreducible representations of a finite group $G$, and let $h : V_{\alpha} \to V_{\beta}$ be a linear map. Define
\begin{equation}
h_0 := \frac{1}{|G|} \sum_{g \in G} \beta(g)^{-1} \, h \, \alpha(g).
\end{equation}
Then:
\begin{enumerate}
    \item If $\alpha \not\cong \beta$, then $h_0 = 0$.
    \item If $\alpha = \beta$ and $V := V_{\alpha} = V_{\beta}$, then $h_0$ is a scalar multiple of the identity:
    \begin{equation}
    h_0 = c I_V,
    \qquad
    c = \frac{1}{\dim V}\operatorname{Tr}(h),
    \end{equation}
    i.e. $h_{0}$ is a homothety of ratio $\tr{h}/{\dim V}$.
\end{enumerate}
\end{restatable}

To use the above we rewrite the following
\begin{align}
    \Pi_{\lambda}\Pi_{2}\Pi_{\lambda}&=\frac{1}{N!}\sum_{{\sigma} \in S_N} \Pi_{\lambda}{\permrev{\hat{P}_{\pi}}{\Pi_{\sigma}}}\Pi_{2}{\permrev{\hat{P}_{\pi}^{\dagger}}{\Pi_{\sigma}^{\dagger}}}\Pi_{\lambda}\\
    &=\frac{1}{N!}\sum_{{\sigma} \in S_N} \Pi_{\lambda}{\permrev{\hat{P}_{\pi}}{\Pi_{\sigma}}}\Pi_{\lambda}\Pi_{2}\Pi_{\lambda}{\permrev{\hat{P}_{\pi}^{\dagger}}{\Pi_{\sigma}^{\dagger}}}\Pi_{\lambda}\\
    &=t_{\lambda}\Pi_{\lambda}
\end{align}
where in the first equation we made use of Lemma \ref{cor: Pi2}, in the second of Lemma \ref{lemma: permutation projection} and lastly by noticing that $ \Pi_{\lambda}\Pi_{2}\Pi_{\lambda}$ is a linear operator that maps an element in the irreducible representation to itself, we can apply the Schur's Lemma. Notice that the value of $t_{\lambda}$ can be computed as follows 
\begin{equation}
    t_{\lambda}=\frac{1}{\chi_{\lambda}(e)\binom{N}{2}}\sum_{1\leq i<j\leq N}\chi_{\lambda}(12)=\frac{\chi_{\lambda}(12)}{\chi_{\lambda}(e)}
\end{equation}
where $\chi_{\lambda}(e)=\operatorname{dim}(\lambda^{S_{N}})$ and  $\chi_{\lambda}(12)$ is the character of the permutations. We can rewrite the character of the permutation for the transposition as follows
\begin{equation}
    \chi_{\lambda}(12)=\frac{\chi_{\lambda}(e)}{\binom{N}{2}}\sum_{i}\left[\binom{\lambda_{i}}{2}-\binom{\overline{\lambda}_{i}}{2}\right]
\end{equation}
see \cite{moshaiov2026polynomial} (Proposition 2.1), with the convention $\binom{1}{2}=\binom{0}{2}=0$. Which leads to
\begin{equation}\label{eq: coefficients irreps}
    t_{\lambda}= \frac{1}{\binom{N}{2}}\sum_{i}\left[\binom{\lambda_{i}}{2}-\binom{\overline{\lambda}_{i}}{2}\right].
\end{equation}

\subsection{Largest eigenvalues of \texorpdfstring{$\Pi_{2}$}{Pi2} outside of the symmetric space}\label{sec: largest eig}

We are interested now to find the largest eigenvalue of $\Pi_{2}$, outside of the symmetric subspace, i.e. $\lambda\neq (N)$. Let us define $\kappa=N-\lambda_{1}$. It is easy to verify that $\lambda\neq (N)$ is equivalent to the condition $\kappa\geq 1$. We start by noticing the following.
\begin{restatable}[]{Lemma}{Majorization}
    Given $\lambda\vdash N$ and $\kappa=N-\lambda_{1}$, we have
\begin{equation}
  \binom{\kappa}{2}\geq  \sum_{i\geq 2}\binom{\lambda_{i}}{2} \ , \ \sum_{i}\binom{\overline{\lambda}_{i}}{2}  \geq \kappa
\end{equation}
\end{restatable}

\begin{proof}
    Let us rewrite the above as a function 
    \begin{equation}
        f(x_{1},...,x_{m})=\sum_{i}\binom{x_{i}}{2}=\frac{1}{2}\sum_{i}x_{i}^{2}-x_{i}
    \end{equation}
    defined for all $\boldsymbol{x}\in \mathbb{N}^{m}: |\boldsymbol{x}|=\kappa $. It is easy to see that the above function is Schur convex\cite{InequalitiesTheoryMajorization}, in particular, it suffices to show that $f$ does not increase under a Robin Hood transfer: if $x_i>x_j$, replace
\begin{equation}
    x_i' = x_i-1,\qquad x_j' = x_j+1,
\end{equation}
leaving the other coordinates unchanged. Then
\begin{align*}
f(\boldsymbol{x}')-f(\boldsymbol{x})
&=\binom{x_i-1}{2}+\binom{x_j+1}{2}-\binom{x_i}{2}-\binom{x_j}{2}\\
&=-(x_i-1)+x_j=x_j-x_i+1\le 0.
\end{align*}
If $f$ is Schur convex, we have that $\boldsymbol{x}\succ \boldsymbol{y}$, i.e. $\boldsymbol{x}$ majorizes $\boldsymbol{y}$, implies $f(\boldsymbol{x})\geq f(\boldsymbol{y})$. Starting by the upper bound, all the partitions with $\kappa>1$ are majorized by $(N-\kappa,\kappa)$, thus the upper bound is 
\begin{equation}
    \binom{N-\kappa}{2}+\binom{\kappa}{2}\geq f(\lambda)
\end{equation}
which implies that the sum starting from the second index, is smaller than $\binom{\kappa}{2}$. We use the same argument for the lower bound, but we need to remember that transposition for partition invert the dominance (majorization) order, meaning that 
\begin{equation}
    \lambda\succ\mu \iff \overline{\mu}\succ \overline{\lambda} .
\end{equation}
So the upper bound, proposed before, works as a lower bound now, in particular the minimum attainable for $f(\overline{\lambda})$ when $\kappa$ is fixed is
\begin{equation}
    f(\overline{\lambda})\geq f((2^{\kappa},1^{N-2\kappa}))=\sum_{i=1}^{\kappa}\binom{2}{2} + \sum_{j=1}^{N-2\kappa}\binom{1}{2}=\kappa
\end{equation}
which concludes the proof.

\end{proof}

If we put the above bound together we can rewrite
\begin{equation}
    t_{\lambda}\leq \frac{\binom{N-\kappa}{2}+\binom{\kappa}{2}-\kappa}{\binom{N}{2}}.
\end{equation}
Let us now consider the case $\kappa=1$, which is achieved only by the partition $\lambda=(N-1,1)$, in particular we have that
\begin{equation}
    \frac{t_{\lambda}-t_{(N-1,1)}}{\binom{N}{2}}\leq \frac{(\kappa-1)(\kappa-N)}{\binom{N}{2}}\leq 0 
\end{equation}
which implies $t_{(N-1,1)}\geq t_{\lambda} \forall \lambda\vdash N: \lambda\neq (N)$. By computing $t_{(N-1,1)}$ directly we conclude that 
\begin{equation}
    t_{\lambda}\leq \frac{N-3}{N-1} .
\end{equation}

\subsection{Largest eigenvalues of \texorpdfstring{$\Pi_{2}$}{Pi2} in the Cyclic invariant subspace}\label{sec: largest eig cycle}

We want now, to consider the previous problem, but restricted to the cyclic subspace. First we consider the case of the partition $\lambda=(N-1,1)$. In particular, one can show that $\Pc|_{(N-1,1)}=0$, which can be derived by realizing that the standard representation of such partition is the space 
\begin{equation}
    V=\{(x_{1},\dots,x_{N})\in \mathbb{C}^{N}:\sum_{i}x_{i}=0\} .
\end{equation}
The only vector in $V$ which is also cyclic symmetric is the vector $\boldsymbol{x}=0$, therefore the he projector onto Cyclic-invariants is the zero operator on this irreducible representation. As a consequence we can reuse the bounds derived before and obtain
\begin{equation}
    t_{\lambda}\leq t_{(N-2,2)}=\frac{N-4}{N} \ \ \forall \lambda\neq \{(N-1,1),(N)\}.
\end{equation}
We now apply the following decomposition 
\begin{align}
    I&=\Psym + (\Pc-\Psym) +(I-\Pc)\\
    &=\mathbb{Q}_{sym}+\mathbb{Q}_{cyc}+\mathbb{Q}_{\perp}
\end{align}
where we relabel the operators for practical reasons, and notice that 
\begin{equation}\label{eq: dominace in subspaces}
    I=\Pi_{2}{\bigg |}_{\mathbb{Q}_{sym}} \ , \ \frac{N-4}{N}I\succeq \Pi_{2}{\bigg |}_{\mathbb{Q}_{cyc}} \ , \ \frac{N-3}{N-1}I\succeq \Pi_{2}{\bigg |}_{\mathbb{Q}_{\perp}} .
\end{equation}

\subsection{Operator bound}\label{sec: operator bound}

For every $|\psi\rangle \in \H_{int}^{\otimes N}$, since $\Pi_{2}$ commutes s with the three mutually orthogonal sector projectors, all cross terms vanish, and we can write 
\begin{align}
    &\langle\psi|\Pi_{2}|\psi\rangle =\sum_{\alpha \in\{ sym,cic,\perp\}} \langle \mathbb{Q}_{\alpha}\psi|\Pi_{2}|\mathbb{Q}_{\alpha}\psi\rangle \\
    &\leq \lVert \mathbb{Q}_{sym}\psi\rVert_{2}^{2}+\frac{N-4}{N}\lVert \mathbb{Q}_{cic}\psi\rVert_{2}^{2} + \frac{N-3}{N-1}\lVert \mathbb{Q}_{\perp}\psi\rVert_{2}^{2}\\
    &=\langle\psi\left|\Psym + \frac{N-4}{N}(\Pc-\Psym) +\frac{N-3}{N-1}(I-\Pc) \right|\psi\rangle .
\end{align}
By defining the quantity
\begin{equation}
    Q_{2}(\rho)=\tr{\Pi_{2}\rho}
\end{equation}
i.e. the average visibility, we can rewrite this equation in terms of expectation values, and rearrange to obtain
\begin{equation}\label{eq: new bound}
    \Find(\rho)\geq \frac{N}{4}\left[Q_{2}(\rho)-\frac{N-4}{N}P_{c}(\rho)-\frac{N-3}{N-1}(1-P_{c}(\rho))\right].
\end{equation}

Since $Q_{2}(\rho)$ must be measured via pairwise HOM, we recall the definition of the Bargmann additive score
\begin{equation}
    A(\rho)=\sum_{k=1}^{N-1}|\tr{{\permrev{\hat{P}_{C^k}^{(int)}}{\Pi_{C^k}^{\mathrm{int}}}} \rho}|^{2/N} .
\end{equation} 
We can notice that for separable states $\rho=\otimes_{i} \rho_{i}$, we have 
\begin{equation}
    Q_{2}(\rho)=\frac{1}{\binom{N}{2}}\sum_{i<j}\tr{\rho_{i}\rho_{j}} .
\end{equation}
We have the following bound.
\begin{restatable}[]{Lemma}{Bound pure}\label{lem: bound p2}
For every $k\in [1,N-1]$ 
\begin{equation}
    |\tr{{\permrev{\hat{P}_{C^k}^{(int)}}{\Pi_{C^k}^{\mathrm{int}}}} \rho}|^{2}\leq \prod_{i=1}^{N}\tr{\rho_{i}\rho_{i+k}}.
\end{equation}
Consequently
\begin{equation}
    Q_{2}(\rho)\geq \frac{1}{N-1}\sum_{k=1}^{N-1} |\tr{{\permrev{\hat{P}_{C^k}^{(int)}}{\Pi_{C^k}^{\mathrm{int}}}} \rho}|^{2/N}=\frac{A(\rho)}{N-1}.
\end{equation}

\end{restatable}
\begin{proof}
    We start by considering a cycle $(i_{1},...,i_{l})$ of $\hat{C}^{k}$ with the notation $i_{l+1}=i_{1}$. Let us define 
    \begin{equation}
        Y_{r}=\sqrt{\rho_{i_{r}}}\sqrt{\rho_{i_{r}+1}} .
    \end{equation}
    We can then rewrite the $N$-th order Bargmann invariant as 
    \begin{equation}
        \tr{\rho_{i_{1}}\dots\rho_{i_l}}=\tr{Y_{1}\dots Y_{l}}.
    \end{equation}
    By using the properties of the Schatten norm we can write the following chain of inequalities
    \begin{align}
        |\tr{Y_{1}\dots Y_{l}}|&\leq \prod_{r=1}^{l}\lVert Y_{r}\rVert_{l}\\
        &\leq \prod_{r=1}^{l}\lVert Y_{r}\rVert_{2}=\prod_{l=1}^{r}\sqrt{\tr{\rho_{i_r}\rho_{i_{r+1}}}} 
    \end{align}
    where we used the fact that $\lVert Y_{r}\rVert_{l}\leq \lVert Y_{r}\rVert_{2} \ \forall l\geq 2$.  Multiplying the above over all the disjoint cycle of $\hat{C}^{k}$, we obtain the inequality
\begin{equation}
    |\tr{{\permrev{\hat{P}_{C^k}^{(int)}}{\Pi_{C^k}^{\mathrm{int}}}} \rho}|^{2}\leq \prod_{i=1}^{N}\tr{\rho_{i}\rho_{i+k}}.
\end{equation}

We can rewrite the above using the AM-GM inequality, to obtain 
\begin{equation}
    \frac{1}{N}\sum_{i=1}^{N}\tr{\rho_{i}\rho_{i+k}}\geq \left(\prod_{i=1}^{N}\tr{\rho_{i}\rho_{i+k}}\right)^{\frac{1}{N}}\geq |\tr{{\permrev{\hat{P}_{C^k}^{(int)}}{\Pi_{C^k}^{\mathrm{int}}}} \rho}|^{\frac{2}{N}}.
\end{equation}
So if we average over the $N-1$ possible values of $k$ we obtain
    \begin{equation}
         Q_{2}(\rho)=\frac{1}{N(N-1)}\sum_{k=1}^{N-1}\sum_{i=1}^{N}\tr{\rho_{i}\rho_{i+k}}\geq \frac{A(\rho)}{N-1}
    \end{equation}
\end{proof}
With this in mind the previous bound can be rewritten in terms of Bargmann score as 
\begin{equation}
    \Find(\rho)\geq \frac{N}{4(N-1)}A(\rho)+\frac{N-2}{2(N-1)}P_{c}(\rho)-\frac{N(N-3)}{4(N-1)}.
\end{equation}
If we define the function 
\begin{equation}\label{eq: Beeest bound}
    G(\rho)=\frac{N}{4(N-1)}A(\rho)+\frac{N-2}{2(N-1)}P_{c}(\rho)-\frac{N(N-3)}{4(N-1)}
\end{equation}
then the optimal lower bound possible with a single Fourier interferometer is
\begin{equation}\label{eq: optimal bound}
    \Find(\rho)\geq \max\left[0,2P_{c}(\rho)-1,G(\rho)\right].
\end{equation}

Additionally we can avoid the use of the function $G(\rho)$, by loosening the strength of the bound and providing a bound depending only on $P_{c}(\rho)$. In particular, we can show the following lemma.
\begin{restatable}[]{Lemma}{}\label{lem: optimal bound in P0}
\begin{equation}
    A(\rho)\geq \frac{N(2P_{c}(\rho)-1)^{\frac{2}{N}}+N-2}{2}
\end{equation}
\end{restatable}
\begin{proof}
    We start by recalling Eq.\eqref{eq: partial result}, where we proved that
    \begin{equation}
        \sum_{k=1}^{N-1}(1-x_{k})\leq N(1-P_{c}(\rho))
    \end{equation}
    with $x_{k}=|\tr{{\permrev{\hat{P}_{C}^{k}}{\Pi_{C}^{k}}}\rho}|$. We can reorder such result to obtain
    \begin{equation}\label{eq: useful for later}
        \frac{1}{N-1}\sum_{k=1}^{N-1}x_{k}\geq \frac{NP_{c}(\rho)-1}{N-1} .
    \end{equation}
    We then notice that $x^{\frac{2}{N}}$ is a concave function, since the exponent is smaller than 1, and thus we have
    \begin{equation}
        x^{\frac{2}{N}}\geq q^{\frac{2}{N}}+\frac{x-q}{1-q}(1-q^{\frac{2}{N}}) \ \ \forall x\in [q,1].
    \end{equation}
    If we substitute $x$ with $x_{k}$ and then average over the $k$, we obtain 
    \begin{align}
        \frac{A(\rho)}{N-1}&\geq q^{\frac{2}{N}}+\frac{\frac{NP_{c}(\rho)-1}{N-1}-q}{1-q}(1-q^{\frac{2}{N}})
    \end{align}
    By substituting $q=2P_C-1$ we find
    \begin{equation}
        A(\rho)\geq \frac{1}{2}\left[N(2P_C-1)^{\frac{2}{N}}+N-2\right].
    \end{equation}
\end{proof}

As a consequence the previous bound reduces to
\begin{equation}\label{eq: bound slightly optimal}
    \Find(\rho)\geq \frac{N\left[N(2P_{c}(\rho)-1)^{\frac{2}{N}}-N+4\right] + 4(N-2)P_{c}(\rho)}{8(N-1)} . 
\end{equation}
We can look at the asymptotic behavior of such quantity, and consider the case $P_{c}(\rho)=1-\epsilon$ with $\epsilon\ll 1$. In this case we have
\begin{equation}\label{eq: Asymptotic}
    1-\epsilon\geq \Find(\rho)\geq 1-\epsilon-\frac{N-2}{2(N-1)}\epsilon^{2}+\mathcal{O}(\epsilon^{3})
\end{equation}
which implies that such a bound is very close to the desired value. 
\subsubsection{Regime of improvement over the elementary cyclic bound}\label{sec: cyclic bound improvement}
Notice that this last bound should be always compared to the original one $2P_{c}(\rho)-1$, but we can show that exception made for the region $P_{c}(\rho)\sim 1/2$ it is always stronger. To show that we rewrite $x=2P_{c}(\rho)-1$, and write the difference between the new bound (Eq.\eqref{eq: bound slightly optimal}) and $x$ to be positive
\begin{equation}
    L_{N}(x)=N^{2}(x^{\frac{2}{N}}-1)-2(3N-2)(x-1)\geq 0
\end{equation}
In the regime $N\leq 5$ the above is always positive for all the value of $x\in [0,1]$, which implies that the bound is always stricter when $P_{c}(\rho)\geq 1/2$. For the case $N\geq 6$ it exists a value $x_{N}$, i.e. the solution $L_{N}(x)=0$, such that $L_{N}(x)\geq 0 \ \forall x\in [x_{N},1]$. We can show that the value $x_{N}$ is an increasing function of $N$, in particular we have
\begin{align}
    \frac{\partial L_{N}(x_{N})}{\partial N} = \frac{\partial 0}{\partial N}=0 \implies \frac{\partial x_{N}}{\partial N}=-\frac{\left.\frac{\partial L}{\partial N}\right|_{x=x_{N}}}{\left.\frac{\partial L}{\partial x}\right|_{x=x_{N}} }.
\end{align}
Differentiating $L_N(x)$ with respect to $x$ yields:
\begin{equation}
    \frac{\partial L_N}{\partial x} = 2N x^{\frac{2-N}{N}} - 2(3N - 2).
\end{equation}
Notice that the above is positive for 
\begin{equation}\label{eq: valid x interval}
    x\in \left[0,\left(\frac{N}{3N-2}\right)^{\frac{N}{N-2}}\right].
\end{equation}
We verify later that the final solution lies in this interval. Differentiating $L_N(x)$ with respect to $N$:
\begin{equation}
    \frac{\partial L_N}{\partial N} = 2N \left( x^{\frac{2}{N}} - 1 \right) - 2 x^{\frac{2}{N}} \ln x - 6(x - 1)
\end{equation}

At the root $x = x_N$, using $L_N(x_N) = 0$, we substitute $x_N^{\frac{2}{N}} - 1 = \frac{2(3N - 2)(x_N - 1)}{N^2}$:
\begin{align}
    \left.\frac{\partial L_N}{\partial N}\right|_{x=x_N} &= (x_N - 1)\left( 6 - \frac{8}{N} \right) - 2 x_N^{\frac{2}{N}} \ln x_N
\end{align}

Using the logarithmic inequality $\ln u < u - 1$ for $u = x_N^{\frac{2}{N}}$, we get:
\begin{equation}
    \ln x_N < \frac{(3N - 2)(x_N - 1)}{N}
\end{equation}

Substituting this bound back into the partial derivative:
\begin{align}
    \left.\frac{\partial L_N}{\partial N}\right|_{x=x_N} &< (x_N - 1)\left( 6 - \frac{8}{N} \right) + \frac{2(3N - 2)(1 - x_N)}{N} x_N^{\frac{2}{N}} \\
    &= (1 - x_N) \left[ -\left(6 - \frac{8}{N}\right) + \left(6 - \frac{4}{N}\right) x_N^{\frac{2}{N}} \right]
\end{align}

Since $x_N \in (0, 1)$, we have $x_N^{\frac{2}{N}} < 1$. Because $\left(6 - \frac{4}{N}\right) > 0$ for $N \ge 6$:
\begin{equation}
    -\left(6 - \frac{8}{N}\right) + \left(6 - \frac{4}{N}\right) x_N^{\frac{2}{N}}  < 0
\end{equation}

Since $(1 - x_N) > 0$ and the term in brackets is strictly negative, it follows that:
\begin{equation}
    \left.\frac{\partial L_N}{\partial N}\right|_{x=x_N} < 0 .
\end{equation}

Thus we conclude that $x_{N}$ is an increasing function of $N$. The maximum value is then reached for $N\to\infty$, which we define as $x_{\infty}=\lim_{N\to \infty}x_{N}$. To compute it we take the limit $L_{N}(x)=0$ and expand in series we have
\begin{align}
    &N^{2}\left[\frac{2}{N}\ln{x_{N}} +\mathcal{O}\left(\frac{1}{N^{2}}\right)\right]=2(3N-2)(x_{N}-1)\implies \\
    &\ln{x_{\infty}}=3(x_{\infty}-1) .
\end{align}
By using the principal branch of the Lambert $W$ we obtain
\begin{equation}
    x_{\infty}=-\frac{1}{3}W_{0}(-3e^{-3})\approx 0.059 \ .
\end{equation}
Notice that this solution lies in the valid interval provided by Eq.\eqref{eq: valid x interval} for all $N\geq 6$. This implies that the bound Eq.\eqref{eq: bound slightly optimal}, is strictly better than $2P_{c}(\rho)-1$ for the values of $P_{c}(\rho)\geq \frac{1+x_{\infty}}{2}\approx 0.5295$ for $N\geq 6$.

\subsection{Finite-data confidence interval from the refined cyclic bound}\label{sec: refined cyclic confidence}

We now apply the $P_c$-only refinement in Eq.~\eqref{eq: bound slightly optimal} to the finite-data setting of Section~\ref{sec: near perfect certification}. For $N\geq4$ and $1/2\leq p\leq1$, define
\begin{equation}\label{eq: refined cyclic lower function}
    \mathcal{B}_N(p)=\frac{N^2(2p-1)^{2/N}-N(N-4)+4(N-2)p}{8(N-1)}.
\end{equation}
For product internal states, Eq.~\eqref{eq: bound slightly optimal} states that $\Find(\rho)\geq \mathcal{B}_N(P_c(\rho))$ whenever $P_c(\rho)\geq1/2$.

\begin{prop}[Refined cyclic confidence interval]\label{prop: refined cyclic confidence}
Let $N\geq4$ and $\rho=\bigotimes_{i=1}^N\rho_i$, with arbitrary single-photon density operators $\rho_i$. Let $\widehat P_c\in[0,1]$ be the empirical estimate from Protocol~\ref{Prot: Fourier}, and suppose that, for fixed $\eta>0$ and $\delta\in(0,1)$,
\begin{equation}\label{eq: refined cyclic estimation event}
    \Pr\!\left(\left|\widehat P_c-P_c(\rho)\right|\leq\eta\right)
    \geq1-\delta.
\end{equation}
Set $a=\widehat P_c-\eta$. With probability at least $1-\delta$, the following data-dependent interval contains $\Find(\rho)$: if $a\geq1/2$, use
\begin{equation}\label{eq: refined cyclic confidence interval}
    \Find(\rho)\in
    \left[\max\{2a-1,\mathcal{B}_N(a)\},\,\widehat P_c+\eta\right];
\end{equation}
otherwise, retain the elementary interval in Eq.~\eqref{eq:cyclic_apriori_fidelity_interval_main}. Intersections with the physical range $[0,1]$ are implicit. In particular, the sample bound in Eq.~\eqref{eq:cyclic_sample_bound_main} suffices without additional measurements or a change in confidence level.
\end{prop}

\begin{proof}
On the event in Eq.~\eqref{eq: refined cyclic estimation event}, we have $a\leq P_c(\rho)\leq\widehat P_c+\eta$. The general upper bound therefore gives $\Find(\rho)\leq\widehat P_c+\eta$, while Theorem~\ref{th: Cyclic theorem} gives $\Find(\rho)\geq2a-1$. For $1/2<p\leq1$,
\begin{equation}
    \mathcal{B}_N'(p)=\frac{N(2p-1)^{2/N-1}+N-2}{2(N-1)}>0.
\end{equation}
By continuity at $p=1/2$, $\mathcal{B}_N$ is increasing on its entire domain. If $a\geq1/2$, the product-state refinement thus implies
\begin{equation}
    \Find(\rho)\geq \mathcal{B}_N(P_c(\rho))\geq \mathcal{B}_N(a).
\end{equation}
Both lower bounds hold on the same estimation event, so taking the larger one, or reverting to the elementary interval when $a<1/2$, incurs no additional failure probability.
\end{proof}

To connect this statement to Eq.~\eqref{eq:near_perfect_cyclic_refinement}, write $a=1-z$, where $z=1-\widehat P_c+\eta\ll1$. Then
\begin{equation}
    \mathcal{B}_N(a)=a-\frac{N-2}{2(N-1)}z^2+\mathcal{O}(z^3).
\end{equation}
Consequently, the width of the refined interval is at most $2\eta+\frac{N-2}{2(N-1)}z^2+\mathcal{O}(z^3)$. Thus the same Fourier data yield a statistical contribution of order $\eta$ and only a second-order residual gap near perfect indistinguishability.

\begin{rem}[Necessity of the product-state assumption]\label{rem: refined bound product assumption}
The nonlinear bound $\mathcal{B}_N(P_c(\rho))$ is not valid for arbitrary separable mixtures. For $N=4$, consider
\begin{equation}
    \rho_{\mathrm{mix}}=\frac12\kb{0000}{0000}+\frac12\kb{0123}{0123},
\end{equation}
where $\ket{0},\ket{1},\ket{2},\ket{3}$ are orthonormal internal states. The two product components have cyclic weights $1$ and $1/4$, and symmetric weights $1$ and $1/24$, respectively. Hence
\begin{equation}
    P_c(\rho_{\mathrm{mix}})=\frac58,\qquad
    \Find(\rho_{\mathrm{mix}})=\frac{25}{48}
    <\mathcal{B}_4\!\left(\frac58\right)=\frac{13}{24}.
\end{equation}
The elementary bound of Theorem~\ref{th: Cyclic theorem} remains valid for this state, but the nonlinear refinement requires the independent-source model stated in Proposition~\ref{prop: refined cyclic confidence}.
\end{rem}

\subsection{Optimality of the Coefficients}\label{sec:optimality}

In this section, we establish the optimality of the coefficients appearing in Eq.~\eqref{eq: dominace in subspaces}. Let $X = \binom{[N]}{2} = \bigl\{\{i,j\} : 1 \leq i < j \leq N\bigr\}$ denote the set of $2$-element subsets of $[N]$, and let $V = \mathbb{C}[X]$ be the corresponding permutation module under the natural action of $S_N$. 

The stabilizer of a point $\{1,2\} \in X$ in $S_N$ is the Young subgroup $S_2 \times S_{N-2}$. Consequently, $X$ is isomorphic as an $S_N$-set to the coset space $S_N / (S_2 \times S_{N-2})$, and $V$ is isomorphic to the Young permutation module 
\begin{equation}
    V \cong M^{(N-2,2)} = \operatorname{Ind}_{S_{N-2} \times S_2}^{S_N} \mathbf{1},
\end{equation}
where $\mathbf{1}$ denotes the trivial representation. By Young's rule, the decomposition of $M^{\mu}$ into Specht modules $S^\lambda$ is governed by the Kostka numbers $K_{\lambda,\mu}$:
\begin{equation}
    M^{\mu} \cong \bigoplus_{\lambda \vdash N} K_{\lambda,\mu} S^{\lambda}.
\end{equation}
For $\mu = (N-2,2)$, the non-zero Kostka numbers are $K_{(N), (N-2,2)} = K_{(N-1,1), (N-2,2)} = K_{(N-2,2), (N-2,2)} = 1$, while $K_{\lambda, (N-2,2)} = 0$ for all other partitions $\lambda$ \cite{fayers2019note}. We thus recover the irreducible decomposition
\begin{equation}\label{eq: decomposition of the transposition space}
    V \cong S^{(N)} \oplus S^{(N-1,1)} \oplus S^{(N-2,2)}.
\end{equation}

To establish optimality, we analyze the action of $\Pi_2$ on each irreducible component of $V$. By Schur's Lemma, $\Pi_2$ acts as a scalar $t_\lambda$ on each Specht module $S^\lambda$ in \eqref{eq: decomposition of the transposition space}, as given by Eq.~\eqref{eq: coefficients irreps}. Direct evaluation yields:
\begin{align}
    t_{(N)}     &= 1, \\
    t_{(N-1,1)} &= \frac{\binom{N-1}{2} - 1}{\binom{N}{2}} = \frac{N-3}{N-1}, \\
    t_{(N-2,2)} &= \frac{\binom{N-2}{2} - 1}{\binom{N}{2}} = \frac{N-4}{N}.
\end{align}
Hence, the lower bounds obtained via upper bounding are tight, proving that no stronger bound can be derived using $\Pc$ and $\Pi_2$ alone.
\section{Tightness of the bound}\label{sec: tightness}

In this section, we will prove that the newly derived upper bound in terms of cyclic projectors are stricter than the previously proposed (Eq.\eqref{eq: P2-bound}) in terms of average visibility. The same cannot be said for the lower bounds, for which we will provide first a counter example and then a condition under which the new bound is tighter.

For the upper bound, we summarize with the following lemma.
\begin{restatable}[]{Lemma}{}\label{lem: Strictness of the bound}
$\forall \rho =\otimes_{i}\rho_{i}\in \mathcal{D}(\H_{int}^{\otimes N})$, we have that 
\begin{align}
    P_{2}(\rho)\geq P_{c}(\rho)\geq \Find(\rho)
\end{align}
\end{restatable}

\begin{proof}
    We recall a previous bound that we derived during the proof of Lemma \ref{lem: bound p2}
    \begin{equation}
        \left|\tr{{\permrev{\hat{P}_{C}^{k}}{\Pi_{C}^{k}}}\rho}\right|\leq \left|\tr{{\permrev{\hat{P}_{C}^{k}}{\Pi_{C}^{k}}}}\right|^{\frac{2}{N}} \leq \frac{1}{N}\sum_{i=1}^{N}\tr{\rho_{i}\rho_{i+k}}.
    \end{equation}
    Then by summing over $k$ we obtain
    \begin{equation}
        \sum_{k=1}^{N-1}|\tr{{\permrev{\hat{P}_{C}^{k}}{\Pi_{C}^{k}}}\rho}|\leq \sum_{k=1}^{N-1}\frac{1}{N}\sum_{i=1}^{N}\tr{\rho_{i}\rho_{i+k}}=(N-1)Q_{2}(\rho)
    \end{equation}
    Thus we have 
    \begin{align}
        P_{c}(\rho)&=\frac{1+\sum_{k=1}^{N-1}\Re{\left\{\tr{{\permrev{\hat{P}_{C}^{k}}{\Pi_{C}^{k}}}\rho}\right\}}}{N}\\
        &\leq \frac{1+\sum_{k=1}^{N-1}\left|\tr{{\permrev{\hat{P}_{C}^{k}}{\Pi_{C}^{k}}}\rho}\right|}{N}\\
        &\leq \frac{1+(N-1)Q_{2}(\rho)}{N}
    \end{align}
    Lastly we notice that the function $f(n)=\frac{1+(n-1)x}{n}$ is a decreasing function of $n$ for $x\in [0,1]$ because
    \begin{equation}
        \frac{\dd}{\dd n}f(n)=-\frac{1-x}{n^{2}}\leq 0 \ \forall x \in [0,1]
    \end{equation}
    which implies 
    \begin{equation}
         P_{2}(\rho)=\frac{1+Q_{2}(\rho)}{2}\geq  \frac{1+(N-1)Q_{2}(\rho)}{N}\geq P_{c}(\rho).
    \end{equation}
    
\end{proof}

Although from the numerical simulations (see Fig.\ref{fig:Comparison}) seems that the newly derived upper bound is stricter it is not always the case. In particular, we can notice that for the state $\rho=|a\rangle\langle a|^{\otimes N-1}\otimes |b\rangle\langle b|$, with $|\langle a|b\rangle|^2=\tau\in [0,1]$, we can compute all the elements to compare the bounds
\begin{align}
    P_{c}(\rho)&=\frac{1+(N-1)\tau}{N}\\
    Q_{2}(\rho)&=1-2\frac{(1-\tau)}{N}\\
    A(\rho)&=(N-1)\tau^{2/N}
\end{align}
which leads to 
\begin{equation}
    \mathcal{L}_{2}(\rho)=(N-1)P_{2}(\rho)-(N-2)\geq G(\rho).
\end{equation}
This implies that  there is no generality on the strictness of the new bound. Nevertheless one can argue that such a counter example is in principle an extremal example, due to the presence of $N-1$ identical particles. Also this example can be regarded as an exceptional case since by performing the calculation for $\Find(\rho)$ we notice that 
\begin{equation}\label{eq: borderline example}
    P_{c}(\rho)=\Find(\rho)=\mathcal{L}_{2}(\rho).
\end{equation}

Based on this we provide a sufficient condition, more grounded in experimental reality, for when the new lower bound is stricter than the one based on $Q_{2}(\rho)$ and for when the opposite happens. To make the notation concise we will refer to $\tr{\rho_{i}\rho_{j}}=v_{ij}$ in the following.

\begin{restatable}[]{Lemma}{}\label{lem: Strictness of the lower bound}
Given $ \rho=\otimes_{i=1}^{N}\rho_{i}\in \mathcal{D}(\H_{int}^{\otimes N})$, let $V_{\min}=\min_{i\neq j}v_{ij}$ and $V_{\max}=\max_{i\neq j}v_{ij}$ and let us define
\begin{align}
    \sigma_c=&\frac{\sum_{k=1}^{N-1}\sum_{i=1}^{N}v_{i,i+k}^{2}}{N(N-1)}-\frac{\sum_{k=1}^{N-1}\left(\sum_{i=1}^{N}v_{i,i+k}\right)^{2}}{N^{2}(N-1)},\\
    \eta(\rho)&=\frac{1}{N-1}\sum_{k=1}^{N-1}\left[\prod_{i=1}^{N}v_{i,i+k}^{\frac{1}{N}}-\left|\tr{{\permrev{\hat{P}_{C}^{k}}{\Pi_{C}^{k}}}\rho}\right|^{2/N}\right].
\end{align} 
Then we have that 
\begin{enumerate}
    \item 
        \(\sigma_{c}\leq \frac{4(N-2)V_{\max}}{N(N-1)}(P_{c}(\rho)-\mathcal{L}_{2}(\rho))-2V_{\max}\eta(\rho)\)
    
    is a sufficient condition for $G(\rho)\geq \mathcal{L}_{2}(\rho)$;
    \item \(
        \sigma_{c}\geq \frac{4(N-2)V_{\min}}{N(N-1)}(P_{c}(\rho)-\mathcal{L}_{2}(\rho))-2V_{\min}\eta(\rho)\)
    is a sufficient condition for $G(\rho)\leq \mathcal{L}_{2}(\rho)$.
\end{enumerate}

\end{restatable}
\begin{proof}
    We start by defining the quantity
    \begin{equation}
        \delta(\rho)=Q_{2}(\rho)-\frac{A(\rho)}{N-1}.
    \end{equation}
    As a consequence we can rewrite 
    \begin{equation}
        G(\rho)-\mathcal{L}_{2}(\rho)=\frac{N-2}{2(N-1)}(P_{c}(\rho)-\mathcal{L}_{2}(\rho))-\frac{N}{4}\delta(\rho).
    \end{equation}
    Let us define the following quantities
    \begin{align}
        &\mu_{k}=\frac{1}{N}\sum_{i=1}^{N}v_{i,i+k} \ , \ s_{k}=\frac{1}{N}\sum_{i=1}^{N}(v_{i,i+k}-\mu_{k})^{2}\ ,\\
        &m_{k}=\min_{i}v_{i,i+k}\ , \
        M_{k}=\max_{i}v_{i,i+k}\ ,\\
        &\eta(\rho)=\frac{1}{N-1}\sum_{k=1}^{N-1}\left[\prod_{i=1}^{N}v_{i,i+k}^{\frac{1}{N}}-\left|\tr{{\permrev{\hat{P}_{C}^{k}}{\Pi_{C}^{k}}}\rho}\right|^{2/N}\right].
    \end{align}
    Notice that 
    \begin{equation}
        \delta(\rho)=\eta(\rho)+\frac{1}{N-1}\sum_{k=1}^{N-1}\left[\mu_{k}-\prod_{i}v_{i,i+k}^{1/N}\right].
    \end{equation}
    The second term appearing in the above sum can be bounded by using the Cartwright-Field inequality\cite{cartwright1978refinement} for each different $k$ as
    \begin{equation}
        \frac{s_{k}}{2M_{k}}\leq \mu_{k}-\prod_{i}v_{i,i+k}^{1/N}\leq \frac{s_{k}}{2m_{k}}.
    \end{equation}
    This implies that 
    \begin{equation}
        \frac{1}{N-1}\sum_{k=1}^{N-1}\frac{s_{k}}{2M_{k}}\leq \delta(\rho)-\eta(\rho)\leq \frac{1}{N-1}\sum_{k=1}^{N-1}\frac{s_{k}}{2m_{k}}.
    \end{equation}
    At this point, we can notice by expanding $\sum_{k}s_{k}$, that 
    \begin{equation}
        \sigma_{c}=\frac{1}{N-1}\sum_{k=1}^{N-1}s_{k}.
    \end{equation}
    By using the fact that $V_{\max}\geq M_{k}$ and $V_{\min}\leq m_{k}$, we obtain
    \begin{equation}
        \frac{\sigma_{c}}{2V_{\max}}\leq \delta(\rho)-\eta(\rho)\leq \frac{\sigma_{c}}{2V_{\min}}.
    \end{equation}
    We have that the condition $G(\rho)\geq \mathcal{L}_{2}(\rho)$ can then be written as 
    \begin{align}
        \frac{2(N-2)}{N(N-1)}(P_{c}(\rho)-\mathcal{L}_{2}(\rho))&\geq \delta(\rho)\\
        & \geq (\delta(\rho)-\eta(\rho))+\eta(\rho).
    \end{align}
    If we substitute the fact that $\delta(\rho)-\eta(\rho) \geq \frac{\sigma_{c}}{2V_{\max}} $, we obtain
    \begin{align}
        \frac{2(N-2)}{N(N-1)}(P_{c}(\rho)-\mathcal{L}_{2}(\rho))&\geq \frac{\sigma_{c}}{2V_{\max}}+\eta(\rho)
    \end{align}
    which once inverted leads to
    \begin{equation}
        \sigma_{c}\leq \frac{4(N-2)V_{\max}}{N(N-1)}(P_{c}(\rho)-\mathcal{L}_{2}(\rho))-2V_{\max}\eta(\rho)
    \end{equation}

    On the other hand the condition $G(\rho)\leq \mathcal{L}_{2}(\rho)$ leads to 
    \begin{align}
        \frac{2(N-2)}{N(N-1)}(P_{c}(\rho)-\mathcal{L}_{2}(\rho))\leq (\delta(\rho)-\eta(\rho))+\eta(\rho).
    \end{align}
    If we substitute the fact that $\delta(\rho)-\eta(\rho) \leq \frac{\sigma_{c}}{2V_{\min}} $, we obtain
    \begin{align}
       \frac{\sigma_{c}}{2V_{\min}}&\geq  \frac{2(N-2)}{N(N-1)}(P_{c}(\rho)-\mathcal{L}_{2}(\rho))-\eta(\rho)
    \end{align}
    which once inverted leads to
    \begin{equation}
        \sigma_{c}\geq \frac{4(N-2)V_{\min}}{N(N-1)}(P_{c}(\rho)-\mathcal{L}_{2}(\rho)) -2V_{\min}\eta(\rho).
    \end{equation}
    
\end{proof}

In essence, $\sigma_c$ quantifies the inhomogeneity among the pairwise visibilities: when the photon sources exhibit a largely homogeneous distribution of overlaps ($\sigma_c \approx 0$), the cyclic Fourier bound $G(\rho)$ is strictly tighter than the average pairwise HOM bound $\mathcal{L}_2(\rho)$, whereas for highly heterogeneous sources with a large visibility spread (such as the presence of a single rogue photon), the arithmetic averaging in $\mathcal{L}_2(\rho)$ becomes more resilient and yields a better lower bound. It is also important to notice that the bound are tighter for pure states, since in that case $\eta(\rho)=0$.

\section{Estimation of \texorpdfstring{$\Find$}{the indistinguishability fidelity} for identical sources without randomization}\label{sec: iid estimation}

When the $N$ photons originate from identical, independent sources, the global internal state is an identically and independently distributed (i.i.d.) product state $\rho = \rho_0^{\otimes N}$, where $\rho_0 \in \mathcal{D}(\H_{\mathrm{int}})$ denotes the single-photon internal state. Such a model is common in atomic boson samplers \cite{young2024atomic,Geller2025}, where the source of partial distinguishability is related to thermal excitations, which can be assumed to be uniform in the systems.

In this setting, the indistinguishability fidelity $\Find(\rho)$ can be determined without implementing the randomized permutation protocol (Protocol~\ref{Prot: Randomized}). In principle, one could quantify the symmetric weight of $\rho = \rho_0^{\otimes N}$ by first reconstructing the spectrum of $\rho_0$ from the power traces $\tr{\rho_0^k}$~\cite{wrightHowLearnQuantum}. However, full spectral reconstruction becomes challenging when the rank of $\rho_0$ is large or continuous (e.g., for broad spectral wave packets). Below, we show that $\Find(\rho)$ can be evaluated directly from the individual cycle traces $\tr{\rho_0^k}$ via a recursive evaluation of the complete Bell polynomials, requiring only fixed Fourier interferometers. We recall that given $\rho_{0}$ such that $\Lambda=(\lambda_{1},...,\lambda_{m})$ is the spectrum of $\rho_{0}$, we have that~\cite{wrightHowLearnQuantum}
\begin{equation}
    \tr{\Psym \rho_{0}^{\otimes N}}=h_{N}\left(\Lambda\right)
\end{equation}
with $h_{k}(\boldsymbol{x})$ the complete homogeneous symmetric polynomial of order $k$.
\begin{restatable}{Theo}{MultipleCopies}\label{theorem: symmetric weight}
Let $\rho = \rho_0^{\otimes N} \in \mathcal{D}(\H_{\mathrm{int}}^{\otimes N})$, for any target accuracy $\varepsilon \in (0,1)$ and confidence parameter $\delta \in (0,1)$, $\Find(\rho)$ can be estimated to additive precision $\varepsilon$ with probability at least $1-\delta$ using a total sample complexity across fixed Fourier interferometers of sizes $k \in \{2, \ldots, N\}$ bounded by
\begin{equation}\label{eq:theorem_iid_sample_complexity}
    M_{\mathrm{tot}} = \mathcal{O}\!\left( \frac{N}{\varepsilon^2} \ln \frac{N}{\delta} \right).
\end{equation}

\end{restatable}

\begin{proof}
By Lemma~\ref{lem:closestIND}, the indistinguishability fidelity is the expectation value of the internal symmetric projector:
\begin{equation}
    \Find(\rho) = \tr{\Psym^{\mathrm{int}} \rho_0^{\otimes N}} = \frac{1}{N!} \sum_{{\sigma} \in S_N} \tr{{\permrev{\hat{P}_\pi^{\mathrm{int}}}{\Pi_\sigma^{\mathrm{int}}}} \rho_0^{\otimes N}}.
\end{equation}
For any integer $n \geq 0$, let $p_s(n) \coloneq \tr{\Psym^{\mathrm{int}} \rho_0^{\otimes n}}$, with $p_s(0) = 1$ and $p_s(N) = \Find(\rho)$. Decomposing permutations into disjoint cycles shows that $p_s(n)$ is generated by the cycle power traces $t_k = \tr{\rho_0^k}$ (where $t_1 = \tr{\rho_0} = 1$) through the generating function
\begin{equation}\label{eq:gen_func_ps}
    H(z) = \sum_{n=0}^\infty p_s(n) z^n = \exp\!\left( \sum_{k=1}^\infty \frac{t_k}{k} z^k \right).
\end{equation}
Differentiating $H(z)$ with respect to a specific trace $t_k$ for $k \geq 2$ yields
\begin{equation}
    \frac{\partial H(z)}{\partial t_k} = \frac{z^k}{k} H(z) = \sum_{m=0}^\infty \frac{p_s(m)}{k} z^{m+k}.
\end{equation}
Equating coefficients of $z^N$ on both sides (setting $m = N-k$), we obtain the exact derivative identity
\begin{equation}\label{eq:exact_partial_derivative}
    \frac{\partial p_s(N)}{\partial t_k} = \frac{1}{k} p_s(N-k).
\end{equation}
Because $p_s(N-k) = \tr{\Psym^{\mathrm{int}} \rho_0^{\otimes (N-k)}}$ is an expectation value of a projector on a valid density matrix, it satisfies $0 \leq p_s(N-k) \leq 1$ over the physical state space. Therefore, for all $k \in \{2, \ldots, N\}$,
\begin{equation}\label{eq:sup_derivative_bound}
    0 \leq \frac{\partial p_s(N)}{\partial t_k} \leq \frac{1}{k}.
\end{equation}

Suppose each trace $t_k$ is estimated by an empirical estimator $\tilde{t}_k$ with additive error bounded by $|t_k - \tilde{t}_k| \leq \eta_k$, setting $\tilde{t}_1 = t_1 = 1$. By the multivariate Mean Value Theorem applied to $p_s(N)$ as a function of $(t_2, \ldots, t_N)$, the error $\Delta F \coloneq |\Find(\rho) - \tilde{p}_s(N)|$ is bounded by
\begin{equation}\label{eq:mvt_error_bound}
    \Delta F \leq \sum_{k=2}^N \left( \sup \frac{\partial p_s(N)}{\partial t_k} \right) |t_k - \tilde{t}_k| \leq \sum_{k=2}^N \frac{\eta_k}{k}.
\end{equation}

To determine the sample allocation, each trace $t_k = \tr{{\permrev{\hat{P}_C^{\mathrm{int}}}{\Pi_C^{\mathrm{int}}}} \rho_0^{\otimes k}}$ is estimated using a fixed $k$-mode Fourier interferometer with photon-number-resolving detection~\cite{novo2026nativelinearopticalprotocolefficient}. Each trial yields a single-shot measurement score $Z^{(k)} \in [-1, 1]$ satisfying $\mathbb{E}[Z^{(k)}] = t_k$. Its variance obeys
\begin{equation}\label{eq:single_shot_variance_k}
    \sigma_k^2 = \operatorname{Var}(Z^{(k)}) \leq 1 - t_k^2.
\end{equation}
By Bernstein's inequality, the empirical mean $\tilde{t}_k$ obtained from $M_k$ independent trials satisfies
\begin{equation}
    \Pr\!\left( |t_k - \tilde{t}_k| \geq \eta_k \right) \leq 2 \exp\!\left( - \frac{M_k \eta_k^2}{2\sigma_k^2 + \frac{4}{3}\eta_k} \right).
\end{equation}
To bound this failure probability by $\delta' = \delta / (N-1)$, it suffices to choose
\begin{equation}\label{eq:bernstein_sample_allocation_k}
    M_k \geq \left[ \frac{2\sigma_k^2}{\eta_k^2} + \frac{4}{3\eta_k} \right] \ln \frac{2(N-1)}{\delta}.
\end{equation}
By the union bound, all $N-1$ estimates simultaneously satisfy $|t_k - \tilde{t}_k| \leq \eta_k$ with probability at least $1-\delta$.

We allocate the individual target accuracies $\eta_k$ to satisfy the total error budget $\sum_{k=2}^N \frac{\eta_k}{k} \leq \varepsilon$ while minimizing the required trials. Choosing $\eta_k = C k^{1/3}$, the total error evaluates to
\begin{equation}
    \sum_{k=2}^N \frac{\eta_k}{k} = C \sum_{k=2}^N k^{-2/3} \leq C \int_0^N x^{-2/3} \dd x = 3 C N^{1/3}.
\end{equation}
Setting $3 C N^{1/3} = \varepsilon$ yields the error allocation
\begin{equation}\label{eq:optimal_eta_choice}
    \eta_k = \frac{\varepsilon}{3} \left( \frac{k}{N} \right)^{1/3}, \qquad k \in \{2, \ldots, N\}.
\end{equation}

For general states, using the uniform bound $\sigma_k^2 \leq 1$, the sum of the variance-dependent terms across all interferometers is
\begin{equation}
    \sum_{k=2}^N \frac{2\sigma_k^2}{\eta_k^2} \leq \frac{18 N^{2/3}}{\varepsilon^2} \sum_{k=2}^N k^{-2/3} \leq \frac{54 N}{\varepsilon^2}.
\end{equation}
Similarly, the linear term sums to
\begin{equation}
    \sum_{k=2}^N \frac{4}{3\eta_k} = \frac{4 N^{1/3}}{\varepsilon} \sum_{k=2}^N k^{-1/3} \leq \frac{6 N}{\varepsilon}.
\end{equation}
Summing Eq.~\eqref{eq:bernstein_sample_allocation_k} over all $k \in \{2, \ldots, N\}$ gives the total sample complexity
\begin{equation}
    M_{\mathrm{tot}} = \sum_{k=2}^N M_k \leq  \mathcal{O}\!\left( \frac{N}{\varepsilon^2} \ln \frac{N}{\delta} \right),
\end{equation}
which proves Eq.~\eqref{eq:theorem_iid_sample_complexity}.

\end{proof}

We emphasize that while the total sample complexity in Eq.~\eqref{eq:theorem_iid_sample_complexity} incurs a mild linear dependence on $N$, this approach yields a direct estimate of the exact fidelity $\Find(\rho)$ rather than a bound. Crucially, it completely bypasses the demanding requirement of shot-to-shot optical reconfiguration inherent to the randomized protocol (Protocol~\ref{Prot: Randomized}), trading the strict $N$-independence of that scheme for a significantly simpler implementation requiring only fixed Fourier networks.

\section{Number of needed Bargmann invariants to predict quantum interference outcomes}\label{Sec: number of BI}
As pointed out in the main text, to predict the outcome of any linear interferometer we need to know an exponential number of Bargmann invariants \cite{Shchesnovich2015,shchesnovich2018collective}. The minimal required number can be computed exactly by noticing that we need to measure all the permutations in the conjugacy classes of the type $(N-k,1^{k})$ for $k\in[0,N-2]$, since the other can be obtained as product of them. We have that  
\begin{equation}
    |Cl(N-k,1^{k})|=\frac{N!}{(N-k)k!}\ \, .
\end{equation}
Notice that except for the case $k=N-2$, all the permutations will also include their inverse, and we need to measure only one of the two, since their respective Bargmann invariant are one the complex conjugate of the other. Thus the total number of invariants we need grows as
\begin{align}
    \operatorname{Inv}(N)&=\binom{N}{2}+\sum_{k=0}^{N-3}\frac{|Cl(N-k,1^{k})|}{2}\\
    &=\frac{1}{2}\binom{N}{2}+\frac{1}{2}\sum_{j=2}^{N}(j-1)!\binom{N}{j}\\
    &=\frac{1}{2}\binom{N}{2}+\frac{N!}{2}\sum_{k=0}^{N-2}\frac{1}{k!(N-k)}\\
    &=\frac{1}{2}\binom{N}{2}+\frac{(N-1)!}{2}\sum_{k=0}^{N-2}\frac{1}{k!(1-\frac{k}{N})}.
\end{align}
If we assume $N\to \infty$ we can rewrite it as 
\begin{align}
    \operatorname{Inv}(N)&\underset{N\to\infty}{\sim}\frac{(N-1)!}{2}\sum_{r\geq 0}\frac{1}{N^{r}}\sum_{k=0}^{\infty}\frac{k^{r}}{k!}\\
    &=\frac{(N-1)!}{2}\sum_{r\geq 0}\frac{1}{N^{r}}eB_{r}\\
    &=\frac{e(N-1)!}{2}\sum_{r\geq 0}\frac{B_{r}}{N^{r}}\\
    &= \frac{e(N-1)!}{2}\left(1+\frac{1}{N}+\frac{2}{N^{2}}+\frac{5}{N^{3}}+\dots\right)
\end{align}
with $B_{r}$ the $r$-th Bell number, for details see \cite{oeisA000110}. For example, for $N=3$ we have $\operatorname{Inv}(3)=4$, meaning the three pair-wise overlaps and the third order invariant. To illustrate the rapid convergence of this asymptotic expansion, Table~\ref{tab:abs_val_comparison_floored} compares the exact number of required Bargmann invariants $\operatorname{Inv}(N)$ against both the leading-order term ($r=0$) and the truncated series up to $r=4$ across various system sizes $N$.
\begin{table}[htbp!]
\centering
\resizebox{\linewidth}{!}{
\begin{tabular}{|c|c|c|c|}
\hline
\rule{0pt}{11pt}
$N$ & $\operatorname{Inv}(N)$ &
$\left\lceil\frac{e(N-1)!}{2}\right\rceil$ &
$\left\lceil\frac{e(N-1)!}{2}\sum_{r=0}^{4}\frac{B_r}{N^r}\right\rceil$ \\
\hline
3   & 4                      & 2                     & 5                     \\
4   & 13                     & 8                     & 12                    \\
5   & 47                     & 32                    & 43                    \\
8   & 8,046                  & 6,850                 & 8,012                 \\
10  & 556,059                & 493,205               & 555,595               \\
15  & $1.28 \times 10^{11}$  & $1.18 \times 10^{11}$ & $1.28 \times 10^{11}$ \\
20  & $1.75 \times 10^{17}$  & $1.65 \times 10^{17}$ & $1.75 \times 10^{17}$ \\
50  & $8.44 \times 10^{62}$  & $8.27 \times 10^{62}$ & $8.44 \times 10^{62}$ \\
100 & $1.28 \times 10^{156}$ & $1.27 \times 10^{156}$ & $1.28 \times 10^{156}$ \\
\hline
\end{tabular}
}

\caption{Comparison between the exact number of Bargmann invariants $\operatorname{Inv}(N)$ required to predict quantum interference outcomes and its asymptotic approximations for various system sizes $N$. Columns 3 and 4 report the leading-order approximation ($r=0$) and the partial sum up to $r=4$, respectively. }
\label{tab:abs_val_comparison_floored}
\end{table}

\section{Negative \texorpdfstring{$c_N$}{cN} coefficient}\label{sec:negative_partition}

In this section, we give an example showing that the genuine multiphoton indistinguishability coefficient \(c_N\) can take an operationally meaningless value, even after permutation twirling. In particular one could consider the example described in \cite{annoniIncoherentBehaviorPartially2025}, in the form $\rho=|\psi\rangle\langle \psi|$ with 
\begin{equation}
    |\psi\rangle=\frac{\ket{0}+\ket{1}}{\sqrt{2}}\otimes \frac{\ket{0}+e^{\frac{2i\pi}{3}}\ket{1}}{\sqrt{2}}\otimes \frac{\ket{0}+e^{\frac{4i\pi}{3}}\ket{1}}{\sqrt{2}}.
\end{equation}

As previously shown this state leads to a coefficient $c_{N}=-\frac{1}{8}\leq 0$, which leads to a trivial lower bound for the Fidelity.

This example is not a peculiarity of the low dimensionality, and can be extended to any $N\geq 3$. More precisely, we consider a family of states for which \(c_N\) is negative for every \(N\geq 3\). After permutation twirling, \(c_N\) can be written as
\begin{align}
c_N &= \frac{1}{N!}\sum_{\pi\in S_N}\tr{\hat{P}_{C}^{(\mathrm{ext})}\hat{P}_{\pi}^{(\mathrm{ext})}\Omega_{\mathrm{ext}}(\rho)\hat{P}_{\pi^{-1}}^{(\mathrm{ext})}} \notag\\
&= \frac{1}{N!}\sum_{\pi\in S_N}\tr{\hat{P}_{\pi^{-1}}^{(\mathrm{ext})}\hat{P}_{C}^{(\mathrm{ext})}\hat{P}_{\pi}^{(\mathrm{ext})}\Omega_{\mathrm{ext}}(\rho)} \\ \notag\
&= \tr{\mathfrak{C}_N\Omega_{\mathrm{ext}}(\rho)},
\end{align}
where we have defined \(\mathfrak{C}_N=\frac{1}{N!}\sum_{\pi\in S_N}\hat{P}_{\pi^{-1}}^{(\mathrm{ext})}\hat{P}_{C}^{(\mathrm{ext})}\hat{P}_{\pi}^{(\mathrm{ext})}\). The operator \(\mathfrak{C}_N\) has a particularly simple decomposition into irreducible permutation sectors.

\begin{restatable}[]{Lemma}{}\label{lem:xi_N}
The operator \(\mathfrak{C}_N\) decomposes as

$$
\mathfrak{C}_N=\sum_{k=0}^{N-1}\frac{(-1)^k}{\binom{N-1}{k}}\Pi_{(N-k,1^k)}.
$$

\end{restatable}

\begin{proof}
The proof follows the same reasoning used previously for the operator \(\Pi_2\). Since \(\mathfrak{C}_N\) is obtained by averaging over conjugations by permutations, it is central in the group algebra. It can therefore be decomposed as \(\mathfrak{C}_N=\sum_{\lambda\vdash N}t_\lambda\Pi_\lambda\). The coefficient associated with the irreducible sector \(\lambda\) is given by the corresponding character ratio, \(t_\lambda=\chi_\lambda(C)/\chi_\lambda(e)\), where \(C\) denotes the conjugacy class appearing in \(\hat{P}_{C}^{(\mathrm{ext})}\).

For the hook partitions \(\lambda=(N-k,1^k)\), with \(k=0,\ldots,N-1\), this character ratio is nonzero and takes the simple form \(t_{(N-k,1^k)}=(-1)^k/\binom{N-1}{k}\). Substituting these coefficients into the central decomposition gives the claimed expression for \(\mathfrak{C}_N\).
\end{proof}

\begin{rem}\label{Rem:finite_size}
If the internal Hilbert space has finite dimension \(d\), i.e. \(\operatorname{dim}\left(\mathcal{H}_{\mathrm{int}}\right)=d\), the sum can be restricted to partitions with at most \(d\) rows. The remaining sectors do not occur in the corresponding Schur--Weyl decomposition and therefore have zero weight.
\end{rem}

We now consider the Oszmaniec--Brod--Galvao (OBG) set of states \cite{Oszmaniec2024}, given by \(\rho=\bigotimes_{j=1}^{N}\psi_j\), with \(|\psi_j\rangle=(|0\rangle+e^{\frac{2\pi i}{N}(j-1)}|1\rangle)/\sqrt{2}\). Since these states are pure, the weight of an irreducible sector can be expressed in terms of the immanant of the corresponding Gram matrix. In particular,
\begin{equation} 
\tr{\Pi_\lambda\Omega_{\mathrm{ext}}(\rho)}=\frac{\chi_\lambda(e)}{N!}\imm_\lambda(G_N),
\end{equation}
where \((G_N)_{ij}=\langle\psi_i|\psi_j\rangle\) is the Gram matrix and \(\imm_\lambda(A)=\sum_{\pi\in S_N}\chi_\lambda(\pi)\prod_{i=1}^{N}A_{i,\pi(i)}\) is the immanant associated with the irreducible character \(\chi_\lambda\). The two extreme cases, \(\lambda=(N)\) and \(\lambda=(1^N)\), correspond to the permanent and determinant, respectively.

The quantities that we need for the OBG family are summarized in the following lemma.

\begin{restatable}[]{Lemma}{}\label{lem:OBG_weights}
For the OBG set of states of size \(N\), the permanent and the \((N-1,1)\)-immanant of the Gram matrix satisfy
\begin{align}
\perm(G_N)&=\frac{N!}{2^{N-1}},\\
\imm_{(N-1,1)}(G_N)&=\frac{N!}{2^{N-1}}\sum_{k=1}^{N-1}\binom{N-1}{k}^{-1}.
\end{align}
\end{restatable}

\begin{proof}
\textbf{1. Permanent of \(G_N\).} By Schur--Weyl duality, the permanent of the Gram matrix can be written as \(\perm(G_N)=N!\langle\Psi|P_{\mathrm{sym}}|\Psi\rangle\), where \(|\Psi\rangle=\bigotimes_{j=1}^{N}|\psi_j\rangle\) and \(P_{\mathrm{sym}}\) is the projector onto the symmetric subspace \(\mathrm{Sym}^N(\mathbb{C}^2)\). This subspace is spanned by the Dicke states \(|D_N^k\rangle=\binom{N}{k}^{-1/2}\sum_{|x|=k}|x\rangle\). For the OBG states, their overlap with a Dicke state is \(\langle D_N^k|\Psi\rangle=\binom{N}{k}^{-1/2}2^{-N/2}e_k(1,\omega,\ldots,\omega^{N-1})\), where \(\omega=e^{2\pi i/N}\).

The generating function \(\prod_{j=0}^{N-1}(1+t\omega^j)=1-(-t)^N\) immediately gives \(e_0=1\), \(e_N=(-1)^{N-1}\), and \(e_k=0\) for \(1\leq k\leq N-1\). Thus only the \(k=0\) and \(k=N\) Dicke states contribute, and \(\langle\Psi|P_{\mathrm{sym}}|\Psi\rangle=2^{-N}+2^{-N}=2^{-(N-1)}\). Hence \(\perm(G_N)=N!/2^{N-1}\).

\medskip
\textbf{2. Immanant \(\imm_{(N-1,1)}(G_N)\).} The character of the standard representation is \(\chi^{(N-1,1)}(\sigma)=\fix(\sigma)-1\). We therefore have
\begin{align}
\imm_{(N-1,1)}(G_N)
&=\sum_{\sigma\in S_N}\bigl(\fix(\sigma)-1\bigr)\prod_{i=1}^{N}(G_N)_{i,\sigma(i)} \notag \\
&=\sum_{k=1}^{N}(G_N)_{kk}\perm(G_N^{[k]})-\perm(G_N),
\end{align}
where \(G_N^{[k]}\) is obtained from \(G_N\) by deleting its \(k\)-th row and column. By the cyclic symmetry of the OBG states and the fact that \((G_N)_{kk}=1\), all of these principal submatrices have the same permanent. Therefore the first term is simply \(N\perm(G_N^{[N]})\).

The matrix \(G_N^{[N]}\) is the Gram matrix of the first \(N-1\) states, corresponding to \(|\Phi\rangle=\bigotimes_{j=1}^{N-1}|\psi_j\rangle\). Expanding \(|\Phi\rangle\) in the \((N-1)\)-qubit Dicke basis, the relevant elementary symmetric polynomials are evaluated on \(\{1,\omega,\ldots,\omega^{N-2}\}\). In this case, \(\prod_{j=0}^{N-2}(1+t\omega^j)=\sum_{k=0}^{N-1}(-\omega^{-1}t)^k\), so all coefficients satisfy \(|e_k|^2=1\). It follows that

$$
\perm(G_N^{[N]})=\frac{(N-1)!}{2^{N-1}}\sum_{k=0}^{N-1}\binom{N-1}{k}^{-1}.
$$

Substituting this result into the expression for the immanant gives
\begin{align}
\imm_{(N-1,1)}(G_N)
&=\frac{N!}{2^{N-1}}\left(\sum_{k=0}^{N-1}\binom{N-1}{k}^{-1}-1\right) \notag\\
&=\frac{N!}{2^{N-1}}\sum_{k=1}^{N-1}\binom{N-1}{k}^{-1}.
\end{align}
\end{proof}

We can now combine the decomposition of \(\mathfrak{C}_N\) with the weights of the OBG states. For this family, the resulting coefficient is

$$
c_N=\frac{1}{2^{N-1}}\left(1-\sum_{k=1}^{N-1}\frac{1}{\binom{N-1}{k}}\right).
$$

For \(N\geq3\), the sum in the parentheses is larger than one, and therefore \(c_N<0\). Thus, the OBG family provides an explicit example in which the coefficient \(c_N\) becomes negative even after permutation twirling. In particular, this shows that \(c_N\) cannot, in general, be interpreted as an operationally meaningful measure of genuine multiphoton indistinguishability.
\clearpage

\end{document}